\documentclass{ps}

 \usepackage[english]{babel}
 \usepackage{multirow}
 \usepackage{amssymb} 
 \usepackage{array} 
 \usepackage{booktabs} 

 \usepackage{amsmath,amssymb,bbm,amsthm }
 \newtheorem{mydef}{Definition}
 \newtheorem{myth}{Theorem}
 
 \usepackage[mathscr]{euscript} 
 \usepackage{accents} 
 \usepackage{todonotes}
  \usepackage[noend]{algorithm2e}

\newcommand{\indic}[1]{\ensuremath{ \mathbbm{1}_{#1} } }

\makeatletter
\newsavebox\myboxA
\newsavebox\myboxB
\newlength\mylenA
\newcommand*\xoverline[2][0.75]{%
 \sbox{\myboxA}{$\m@th#2$}%
 \setbox\myboxB\null
 \ht\myboxB=\ht\myboxA%
 \dp\myboxB=\dp\myboxA%
 \wd\myboxB=#1\wd\myboxA
 \sbox\myboxB{$\m@th\overline{\copy\myboxB}$}
 \setlength\mylenA{\the\wd\myboxA}
 \addtolength\mylenA{-\the\wd\myboxB}%
 \ifdim\wd\myboxB<\wd\myboxA%
 \rlap{\hskip 0.5\mylenA\usebox\myboxB}{\usebox\myboxA}%
 \else
 \hskip -0.5\mylenA\rlap{\usebox\myboxA}{\hskip 0.5\mylenA\usebox\myboxB}%
 \fi}
\makeatother

\numberwithin{equation}{section}
\begin{document}

\title{On the optimal importance process for piecewise deterministic Markov process} 

\author{H. Chraibi}\address{PERICLES department, EDF lab saclay, 7 Bd Gaspard Monge, 91120 Palaiseau, France}
\author{A. Dutfoy}\address{PERICLES department, EDF lab saclay, 7 Bd Gaspard Monge, 91120 Palaiseau, France}
\author{\underline{T. Galtier}}\address{LPSM (Laboratoire de probabilit\'es, statistique et mod\'elisation), Universit\'e Paris Diderot, 75205 Paris Cedex 13, France,tgaltier@gmail.com}
\author{J. Garnier}\address{CMAP (Centre de math\'ematiques appliqu\'ees), \'Ecole polytechnique, 91128 Palaiseau Cedex, France}

\begin{abstract}
In order to assess the reliability of a complex industrial system by simulation, and in reasonable time, variance reduction methods such as importance sampling can be used. We propose an adaptation of this method for a class of multi-component dynamical systems which are modeled by piecewise deterministic Markovian processes (PDMP). We show how to adapt the importance sampling method to PDMP, by introducing a reference measure on the trajectory space. This reference measure makes it possible to identify the admissible importance processes. Then we derive the characteristics of an optimal importance process, and present a convenient and explicit way to build an importance process based on theses characteristics. A simulation study compares our importance sampling method to the crude Monte-Carlo method on a three-component systems. The variance reduction obtained in the simulation study is quite spectacular.\\
\end{abstract}


\maketitle
\section{Introduction}
\label{sec1}
For safety and regulatory issues, nuclear or hydraulic industries must assess the reliability of their power generation systems. To do so, they can resort to probabilistic safety assessment. In recent years, dynamic reliability methods have been gaining interest, because they avoid conservative static approximations of the systems and they better capture the dynamics involved in the systems. When dealing with complex industrial systems, this kind of reliability analysis faces two main challenges: the first challenge is related to the modeling of such complex systems, the second one concerns the quantification of the reliability. Indeed as we refine the model the estimation of the reliability requires more efforts and is often challenging. 

\subsection{A model based on a PDMP}
\label{subsec:PDMP} 
Due to the complexity of the systems the reliability analysis is often done through an event tree analysis \cite{cepin2011assessment} which requires static approximations of the system, and  relies on conservative approximations. With the development of computational capacities, it is now possible to consider more accurate tools for reliability  assessment. Several attempts have been proposed to better model the dynamical processes involved in the systems. In this article, we focus on the option proposed in  \cite{zhang2008piecewise} and \cite{dufour2015numerical}, which consists in modeling the system using a piecewise deterministic Markovian process (PDMP) with boundaries.

In many industrial systems, and in particular in power generation systems, failure corresponds to a physical variable of the system (such as temperature, pressure, water level) entering a critical region. The physical variables can enter this region only if a sufficient number of the basic components of the system are damaged. In order to estimate the reliability we need an accurate model of the trajectories of the physical variables. In industrial systems, the physics of the system is often determined by ordinary or partial differential equations which depend on the statuses of the components within the system (on, off or failed). Therefore the dynamics of the physical variables changes whenever the statuses of the components are altered. Such alteration can be caused by automatic control mechanisms within the systems or failures or repairs. It is also possible that the values of physical variables impact the statuses of the components, because the failure and repair rates of the components depend on the physical conditions. In order to deal with this interplay between the physical variables and the statuses of components, we need to model their joint evolution. The vector gathering these variables is called the state of the system. To address the challenge of modeling the trajectory of the state of the system,  we model the evolution of the state of the system by a piecewise deterministic Markovian process (PDMP) with boundaries. PDMPs were introduced by M.H.A Davis in \cite{davis1984,davis1993markov}, they benefit from high modeling capacity, as they are meant to represent the largest class of Markovian processes that do not include diffusion. These processes can easily incorporate component aging, failure on demand, and delays before repairs. 
 
For a given system, we denote its state at time $t$ by $Z_t =(X_t,M_t)$, where $X_t$ is the vector of the values of the physical variables, and $M_t$ the vector gathering the statuses of all the components in the system. Throughout the paper we call $X_t$ the position of the system, and $M_t$ the mode of the system. $\mathbf Z=(Z_t)_{t\in[0,t_f)}$ represents a trajectory of the state of the system up to a final observation time~$t_f$. We consider that the trajectories are all initiated in a state $z_o$.

Recall the system fails when the physical variables enter a critical region. We denote by $D$ the corresponding region of the state space, and we denote by $\mathscr D$ the set of the trajectories of $\mathbf Z$ that pass through $D$. In order to estimate the reliability on the observation time $t_f$, we want to estimate the probability of system failure defined by $$p=\mathbb{P}\big(\mathbf{Z}\in \mathscr D |Z_0\mbox{\hspace{-0.3 ex}}=\mbox{\hspace{-0.35ex}}z_o\big)=\mathbb{P}_{z_o}\big(\mathbf{Z}\in \mathscr D \big).$$


\subsection{Accelerate reliability assessment by using importance sampling}
\label{subsec:Pyc}
The second challenge is that the reliability of a complex industrial system can rarely be assessed analytically, so reliability analysis often relies on simulations techniques. The company \'Electricit\'e de France (EDF) has recently developed the PyCATSHOO toolbox \cite{Pycatshoo,Pycatshoo2}, which allows the simulation and the modeling of dynamic hybrid systems. PyCATSHOO bases its modeling on PDMPs. Thanks to Monte-Carlo simulation, it evaluates dependability criteria, among which is the reliability of the system. The method we present in this article is used to accelerate the reliability assessment within the PyCATSHOO toolbox.

In the context of reliable systems, crude Monte-Carlo techniques perform poorly because the system failure is a rare event. Indeed, with the Monte-Carlo method, when the probability of failure approaches zero, the number of simulations to get a reasonable precision on the relative error increases dramatically, and so does the computational time. To reduce this computational burden, one option is to reduce the number of simulations needed by using a variance reduction method. Among variance reduction techniques \cite{caron2014some,morio2014survey}, we may think of multilevel splitting techniques \cite{cerou2006genetic, del2005genealogical} and of importance sampling techniques \cite{dufour2015numerical,tutorialCE,bucklew2013introduction,zio2013monte}. A variance reduction method, inspired from particle filtering can be used on a particular case of PDMP that is a PDMP whithout boundary \cite{whiteley2011monte}. Unfortunately the industrial systems are often modeled by a PDMP with boundaries, and other variance reduction methods need  to be designed for these cases. We choose to focus on the importance sampling technique, because: 1) the importance sampling strategy that we propose can easily be implemented (in particular in the PyCATSHOO toolbox) 2) the results derived in this paper (in particular the reference measure and the expressions of the densities and likelihood ratios) should be useful to study multilevel splitting.

In this paper we present how to adapt the importance sampling technique for PDMP. By doing so we generalize the use of importance sampling, not only for many power generation systems, but also for any phenomenon that can be modeled by a PDMP. As PDMP generalizes numerous kinds of processes (among which are discrete Markov chains, continuous time Markov chains, compound Poisson processes or queuing systems), the scope of our work goes way beyond the study of power generation systems.

\subsubsection{Prerequisite for importance sampling on PDMPs}
\label{subsec:IS}
Remember that we want to apply importance sampling to estimate the probability $p=\mathbb{P}_{z_o}\big(\mathbf{Z}\in \mathscr D \big)$ that the system fails. In our case, importance sampling would consist in simulating from a more fragile system, while weighting the simulation outputs by the appropriate likelihood ratio. The issue is that the random variable we are considering is a trajectory of a PDMP, so we need to clarify what is the density (or the likelihood) for a trajectory of PDMP. Namely we need to introduce a reference measure for PDMP trajectories, and to identify its related densities.

In simple cases of dynamical importance sampling, this issue of the reference measure is often eluded, because the reference measure has an obvious form: it is often a product of Lebesgue measures, or a product of discrete measures. But PDMPs are very degenerate processes, their laws involve hybrid random variables which have continuous and discrete parts. In this context, it is important to ensure that we do have a reference measure that is sigma-finite to define properly the densities and the likelihood ratios.

Suppose $\zeta$ is a reference measure for $ \mathbb{P}_{z_o}\big(\mathbf{Z} \in . \big)$, we denote by
 $f$ the density of $\mathbf Z$ with respect to $\zeta$, and we denote by
 $g$ the density of an importance process $\mathbf Z^\prime$ with respect to $\zeta$. If $\zeta$ exists, and $f$ and $g$ satisfy $\forall \mathbf z\in\mathscr D,\ f(\mathbf{z})\ne 0\Rightarrow g(\mathbf{z})\ne 0,\, $ then we can write:
 \begin{align}
 \mathbb{P}_{z_o}\big(\mathbf{Z}\in \mathscr D \big) = \mathbb{E}_f\big[\indic{\mathscr D}(\mathbf Z)\big] 
 &= \int_\mathscr{D} f(\mathbf{z})\,d\zeta(\mathbf{z}) = \int_\mathscr{D}\dfrac{f(\mathbf{z})}{g(\mathbf{z})}g(\mathbf{z})\,d\zeta(\mathbf{z}) = \mathbb{E}_g\bigg[\indic{\mathscr D}(\mathbf Z)\dfrac{f(\mathbf{Z})}{g(\mathbf{Z})}\bigg]\label{eq:ISPDMP}.
 \end{align}
 If $\big(\mathbf Z_1^\prime,\dots \mathbf Z_{N_{sim}}^\prime\big)$ is a sample of independent trajectories simulated according to an importance process with density $g$, then $\mathbb{P}_{z_o}\big(\mathbf{Z}\in \mathscr D \big) $ can be estimated without bias by:
 \begin{align}
 \hat{p}_{IS}&= \frac 1 {N_{sim}} \sum_{i=1}^{N_{sim} }\indic{ \mathscr D }(\mathbf Z_i^\prime) \frac{f(\mathbf Z_i^\prime)}{g(\mathbf Z_i^\prime)} & \mbox{with}\quad \mathbb V\mbox{ar}(\hat{p}_{IS}) =\frac {\mathbb E_f\left[\indic{\mathscr D}(\mathbf Z)\frac{f(\mathbf Z)}{g(\mathbf Z)}\right]-p^2 }{N_{sim}} 
 \label{pIS}
 \end{align}
 When $\mathbb E_{f}\big[\indic{ \mathscr D }(\mathbf Z) \frac{f(\mathbf Z)}{g(\mathbf Z)}\big] <\infty$ and the conditions above are verified, we have a central limit theorem on $\hat p_{IS}$:
 \begin{equation}
 \sqrt{N_{sim}} (\hat p_{IS}-p)\longrightarrow \mathcal N(0,\sigma_{IS}^2)\quad \mbox{where}\quad \sigma_{IS}^2=\mathbb E_{f}\big[\indic{ \mathscr D }(\mathbf Z) \frac{f(\mathbf Z)}{g(\mathbf Z)}\big] - p^2.
 \end{equation}

 Thus the use of importance sampling on PDMP trajectories requires the following three conditions:
 \begin{itemize} 
 \item[(C1)] We have a measure $\zeta$ on the trajectory space, and the trajectory $\mathbf Z$ of the system state has density $f$ with respect to $\zeta$ 
 \item[(C2)] We are able to simulate trajectories according to an importance process $\mathbf Z^{\prime} $ which has density $g$ with respect to $\zeta$ on $\mathscr D $ such that $\mathbb E_{f}\big[\indic{ \mathscr D }(\mathbf Z) \frac{f(\mathbf Z)}{g(\mathbf Z)}\big] <\infty$.
 \item[(C3)]  $\zeta$-almost everywhere in $ \mathscr D$ we have $\ f(\mathbf{z} )\ne 0 \Rightarrow g(\mathbf{z} )\ne 0$  
 \end{itemize} The existence of a reference measure is an important theoretical argument, but it can also be used to characterize the admissible importance processes. Knowing the reference measure $\zeta$ tells us what we can modify in the law of $\mathbf Z$ to obtain an importance process $\mathbf Z^\prime$ with a well-defined likelihood ratio. It is a valuable information to know to which extent we can modify the density $f$ to get the density $g$, because the variance $\sigma_{IS}^2$ depends on the density $g$.
 
 It is theoretically possible to design an importance sampling strategy with zero variance, indeed, it suffices to use an importance process with a density \begin{equation}
     g^*(\mathbf z)=\dfrac{\indic{\mathscr D}(\mathbf z)}{p}f(\mathbf z)\ .
 \end{equation} In practice, however, we cannot reach this zero variance, as we do not know the value of $p$. The expression of $g^*$ rather serves as a guide to build an efficient and explicit density $g$. Indeed we can try to choose a density $g$ as close as possible from $g^*$ in order to get a strong variance reduction. \\

\subsubsection{Our contributions to the literature}
 Many authors have used importance sampling on particular cases of PDMP sometimes without noting it was PDMPs, see \cite{labeau1996a, labeau1996b, lewis1984, Marseguerra1996}. Sometimes, the authors using PDMPs avoid considering automatic control mechanisms which activate and deactivate components depending on the values of physical variables. Such automatic control mechanisms play an important part in power generation systems, and therefore that can not be avoided in our case. Also, the modeling of control mechanisms implies to work with a special kind of PDMPs, which are the PDMPs with boundaries. These PDMP are typically the kind for which the reference measure is complex. In \cite{Marseguerra1996,ramakrishnan2016unavailability}, importance sampling is used on PDMP while taking into account automatic control mechanisms but the reference measure is not clearly identified, and so far we have not found a proof that likelihood ratios involved in importance sampling on PDMP are always defined. In Section \ref{sec:zeta}  we provide a reference measure for PDMP trajectories. This allows to define  the likelihood ratios for PDMP trajectories and to use the importance sampling method, but also to identify the admissible importance processes. Our major contributions are presented in Section \ref{sec:Opti} where  the characteristics of the optimal importance process are identified, and used to  propose a convenient way to build the importance process in practice. Note that the characteristics of the optimal importance process are identified for the general case of  PDMP, therefore our result can be generalized to any subclass of the PDMP process, like Markov chains, or  continuous time Markov chains, or queuing systems.

\subsubsection{Optimization of the variance reduction}
\label{subsec:CE}

Finding the optimal importance process is equivalent to solving the following minimization problem: $$g^{*}=\underset{g}{\mbox{argmin }} \mathbb E_{f}\big[\indic{ \mathscr D }(\mathbf Z) \frac{f(\mathbf Z)}{g(\mathbf Z)}\big]$$ Minimizing a quantity on a density space being difficult, we usually consider a parametric family
of importance densities $\{g_\alpha \}$ and look for a parameter $\alpha$ which yields an estimator with the smallest possible variance. Under favorable circumstances the form of the family can be determined by a large deviation analysis \cite{dupuis2004importance, heidelberger1995fast, siegmund1976importance}, but the large deviation method is difficult to adapt to PDMP with boundaries which are degenerate processes with state spaces with complicated topologies. Therefore we focus on other methods which rather try to minimize an approximation of the distance between the importance density $g$ and the optimal one $g^{*}$. For instance, if the approximated distance happens to be $ D(g,g^{*}) = \mathbb E_{f}\big[ \frac{g^{*}(\mathbf Z)}{g(\mathbf Z)}\big] $ it is equivalent to minimize the variance of the estimator, and if we consider the Kullback-Leibler divergence so that $D(g,g^{*}) = \mathbb E_{g^{*}}\Big[\log\big( \frac{g^{*}(\mathbf Z)}{g(\mathbf Z)}\big)\Big] $, we would be using the Cross-Entropy method \cite{tutorialCE, zio2013monte}. These two options have been compared on a set of standard cases in \cite{chan2011comparison}. They yielded similar results, though results obtained with the Cross-Entropy seemed slightly more stable than with the other option. In \cite{zuliani2012rare} the Cross-Entropy method was applied to a model equivalent to a PDMP without boundaries and showed good efficiency. Therefore we choose this method to select the parameters of the importance process in our paper. Of course, the efficiency of this procedure strongly depends on the choice of the parametric family of importance densities.\\

The rest of the paper is organized as follows: Section~\ref{sec:PDMP} introduces our model of multi-component system based on a Piecewise deterministic Markovian process. In Section~\ref{sec:zeta}, we introduce a reference measure on the space of the PDMP trajectories and study the admissible importance processes. In Section~\ref{sec:Opti} we present an optimal process and a clever way to build the importance process in practice. In Section~\ref{sec:Ex} we apply our adaptation of the importance sampling technique on a three-component system and compare its efficiency with the Monte-Carlo technique.\\

\section{A model for multi-component systems based on PDMP}
\label{sec:PDMP}

 \subsection{State space of the system}
 
 We consider a system with $N_c$ components and $d$ physical variables. Remember we call position the vector $X\in \mathbb R^d$ which represents the physical variables of the system, and we call mode the vector $M=(M^1,M^2, ..., M^{N_c})$ gathering the statuses of the $N_c$ components. The state of the system $Z$ includes the position and the mode: $Z=(X,M)$.\\
 
 For ease of the presentation, we consider the status of a component can be   $ON$, or $OFF$, or out-of-order (noted $F$), so that the set of modes is $\mathbb M=\{ON,OFF,F\}^{N_c}$, but as long as $\mathbb M$ stays countable, it is possible to consider more options for the statuses of the components. For instance, one could consider different regimes of activity instead of the simple status $ON$, or different types of failure instead of the status $F$. Note that we can also deal with continuous degradations, like the size of a breach in a pipe for instance: the presence of the degradation can be included in the mode and its size in the position.\\
 
 Generally, there are some components in the system which are programmed to activate or deactivate when the position crosses some thresholds. For instance, it is typically what happens with a safety valve: when the pressure rises above a safety limit, the valves opens. To take into account these automatic control mechanisms, within a mode $m$ the physical variables are restricted to a  set $\Omega_m\subset \mathbb R^d$, which is assumed open. We set $E_m=\{(x,m), x\in\Omega_m\}$, so that the state space is: \begin{equation}E=\underset{m\in\mathbb{M}}{\bigcup} E_m =\underset{m\in\mathbb{M}}{\bigcup} \Big\{(x,m), x\in\Omega_m\Big\} \end{equation}
 
 \subsection{Flow functions}
 
 In a given mode $m$, i.e. a given combination of statuses of components, the evolution of the position is determined by an ordinary differential equation. We denote by $\phi^m_x $ the solution of that equation initiated in $x$. If we consider a position state $Z_t$ at time $t$, there exists a random time $T>0$ such that $\forall s\in[0,T), $ $X_{t+s}=\phi^{M_t}_{X_t}(s)$ and $M_{t+s}=M_t$. For an initial state $z\in E$, we can introduce the flow function $\Phi_z$ with values in $E$. Regarding the evolution of the trajectory after a state $Z_t=(X_t,M_t)$, the next states are locally given by the function $\Phi_{Z_t}$:
 \begin{align}
&\exists T>0,\,\forall s\in[0,T) ,\quad \nonumber\\
& Z_{t+s} = \Phi_{Z_t}(s)= \big(\phi^{M_t}_{X_t} (s),M_t\big) = \big( X_{t+s},M_t\big) \label{flow}
 \end{align}
In practice an approximation of the function $\phi^{m}_{x}$ can be obtained by using a numerical method solving the ordinary differential equations. For instance the PyCATSHOO toolbox can use, among others, the Runge-Kutta methods up to the fourth order.

 \subsection{Jumps}
 \label{subsec:jumps}
 The trajectory of the state can also evolve by jumping. This typically happens because of control mechanisms, failures, repairs, or natural discontinuities in the physical variables. When such a jump is triggered, the current state moves to another one by changing its mode and/or its position. 
We denote by $\xoverline{ E\,}$ the closure of $E$, and $\mathscr B(E)$ the Borelian $\sigma$-algebra on $E$. We define $T$ as the time until the next jump after $t$, such that the next jump occurs at time $t+T$. The destination of this jump is determined according to a transition kernel $\mathcal K_{Z_{t+T^-}}$ where $Z_{t+T^-}\in \xoverline{ E\,}$ is the departure state of the jump. This kernel is defined by:  \begin{align}
&\forall B\in \mathscr B (E),\quad  
\mathbb P \left(Z_{t+T} \in B|Z_{t+T^-}=z^-\right) =\mathcal K_{z^-}(B) \  ,\label{kernel0} 
 \end{align} where $\mathscr B(.)$ indicates the Borelians of a set. Let $\forall z^-\in\xoverline{ E\,},\,\nu_{z^-}$ be a $\sigma$-finite measure on $E$, such that $\mathcal K_{z^-}<<\nu_{z^-}$. $\forall z^-\in\xoverline{ E\,}$, we define  $K_{z^-}$  as the density of $\mathcal K_{z^-}$ with respect to $\nu_{z^-}$, so:
 \begin{align}
&\forall B\in \mathscr B (E),\quad  
\mathbb P \left(Z_{t+T} \in B|Z_{t+T^-}=z^-\right) =\int_B K_{z^-}(z)\,d \nu_{z^-}(z) \  ,\label{kernel} 
 \end{align}
   The kernel density must satisfy $K_{z^-}(z^-)=0$, so that even  if $\nu_{z^-}$ has a Dirac point in $z^-$,   jumping on the departure state is impossible. Note that with this setting, the law of the arrival state of a jump can depend on the departure point $z^-$.  For instance, if the physical variables are all continuous, then the reference measure of the transition kernel $\nu_{z^-}$ could be defined by:  \begin{align}\forall B\in\mathscr B (E), \qquad &  \nu_{z^-} (B)=\sum_{
  \begin{array}{c}
       w\in\mathbb{M}\backslash\{m^-\},\\
       (x^-,w)\in E
  \end{array}
  } \delta_{(x^-,w)}(B) . \end{align}
 In this example, the jump kernel is discrete: 
  \begin{align}\forall B\in\mathscr B (E), \qquad &  \mathcal K_{z^-} (B)=\sum_{
  \begin{array}{c}
       w\in\mathbb{M}\backslash\{m^-\},\\
       (x^-,w)\in E 
  \end{array} }
  \mathbb P \left(Z_{t+T}=(x^-,w)|Z_{t+T^-}=(x^-,m^-)\right)\delta_{(x^-,w)}(B) , \end{align}
 and it is generally the case, but one can imagine some cases where the kernel  include a continuous part.  For instance consider that the physical variables have two dimensions, the first corresponding to pressure on a concrete structure, and the second to the size of a crack in the structure. One can consider that the crack length increase in a jerky way, and that the amplitude of the increase  is random and has a continuous law. For a jump triggered from a state $z^-=\big((x_1^-,x_2^-),m^-\big)\in E$ we could have:
 \begin{align}\forall B\in\mathscr B (E), \qquad &  \nu_{z^-} (B)= \int_{\big\{y_2>0\big|\big((x_1^-,y_2),m^-\big)\in B\big\}}d\mu_{Leb}(y_2) \end{align}
 where $\mu_{Leb}(.)$ corresponds to the Lebesgue measure, and  \begin{align}\forall B\in\mathscr B (E), \qquad &  \mathcal K_{z^-} (B)=\int_B K_{z^-}(z)\,d \nu_{z^-}(z) =\int_{\big\{y_2>0\big|\big((x_1^-,y_2),m^-\big)\in B\big\}}K_{z^-}\big(\big((x_1^-,y_2),m^-\big)\big)d\mu_{Leb}(y_2) \end{align} 
 with $K_{z^-}\big(\big((x_1,x_2),m\big)\big)=0$.  We think, that the cases of non-discrete jump kernel should be rather rare in the reliability analysis field, but PDMPs are also used in other fields, like finance, where  non-discrete jump kernel could be more common and for which the use of importance sampling can be of interest too \cite{rolski2009stochastic}. That is why we keep the most general form of PDMP, which can handle any type of jump kernel.
 
 \subsection{Jump times}
 Now, assuming that $Z_t=z$, we present the law of the time until the next jump after $t$, which is denoted by~$T$.
 \subsubsection*{Jumps at boundaries}
 For $m\in \mathbb M$, let $\partial\Omega_m$ be the boundary of $\Omega_m$. The boundary of the set $E_m$ is the set $\partial E_m =\{(x,m), x\in\partial\Omega_m\} $. For $z=(x,m)\in E$, we define $t^*_z=\inf\{s>0,\Phi_z(s)\in \partial E_m\}$ the time until the flow hits the boundary. We take the convention $t^*_z= + \infty\ $ if $\{s>0,\Phi_z(s)\notin E_m\}=\emptyset$.
Assume that the system starts in state $z=(x,m)\in E$. When the flow leads the position out of its restricted set $ \Omega_m$, i.e. the state touches $\partial E_m$, an automatic jump is triggered (see the scheme in \ref{sautdet}), and $T=t^*_z$.

\begin{figure}[ht]\centering
\includegraphics[width= 0.9\linewidth]{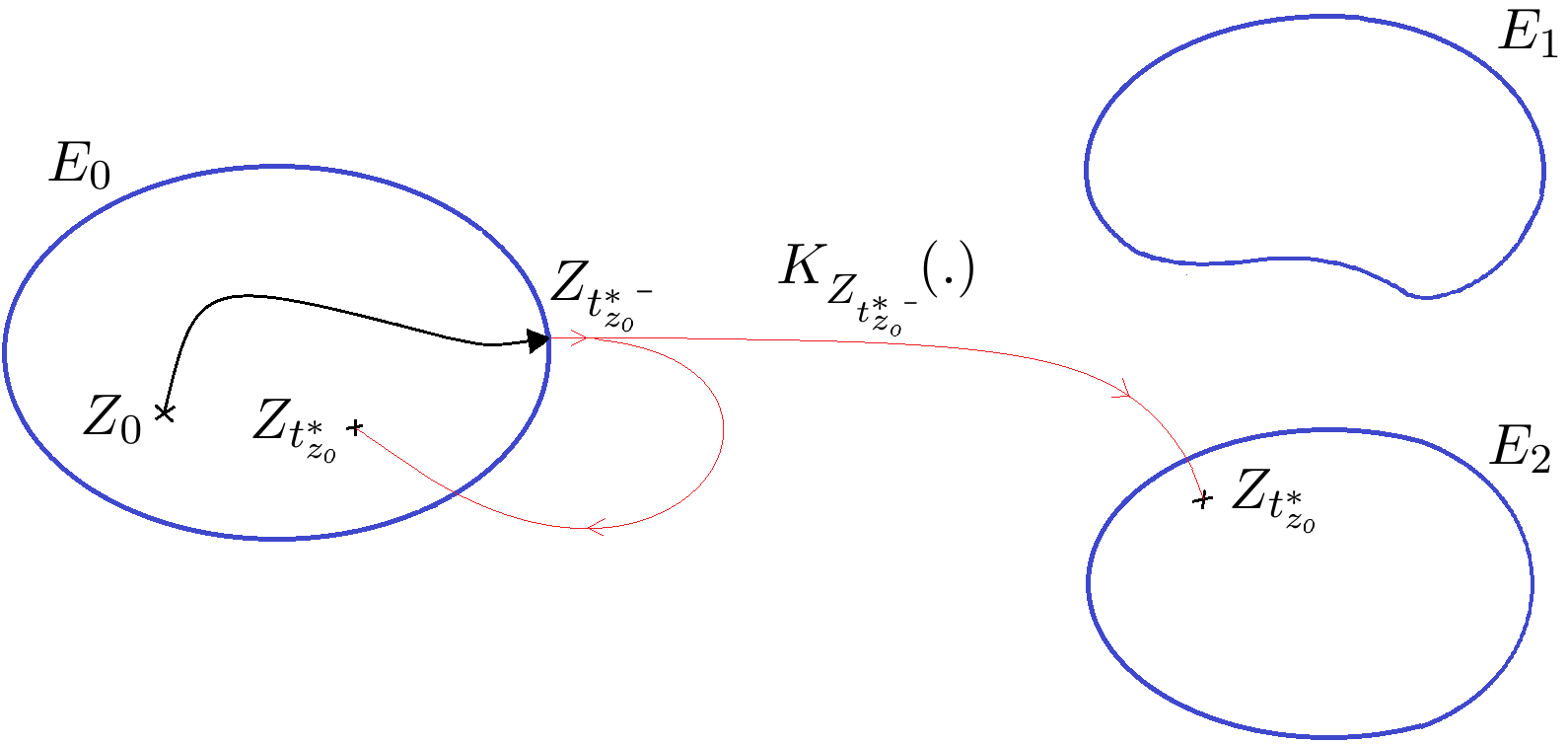}
\caption{A jump at boundary.\label{sautdet}} 
\end{figure}

Boundaries can be used to model automatic control mechanisms, or any automatic change in the status of a component. For instance in a dam, if the water level $X$ reaches a given threshold $x_{max}$ the evacuation valve automatically opens to avoid overflow. If $M= C ,\ O ,\ F $ represent respectively the modes where the valve is closed, or opened, or failed, this control system could be modeled by setting $\Omega_{C}=(0,x_{max})$ and $K_{(x_{max},C)}(\{(x_{max},O)\})=1$.\\

 \subsubsection*{Spontaneous jumps}
 \begin{figure}[hb] \centering
 \includegraphics[width=0.9\linewidth]{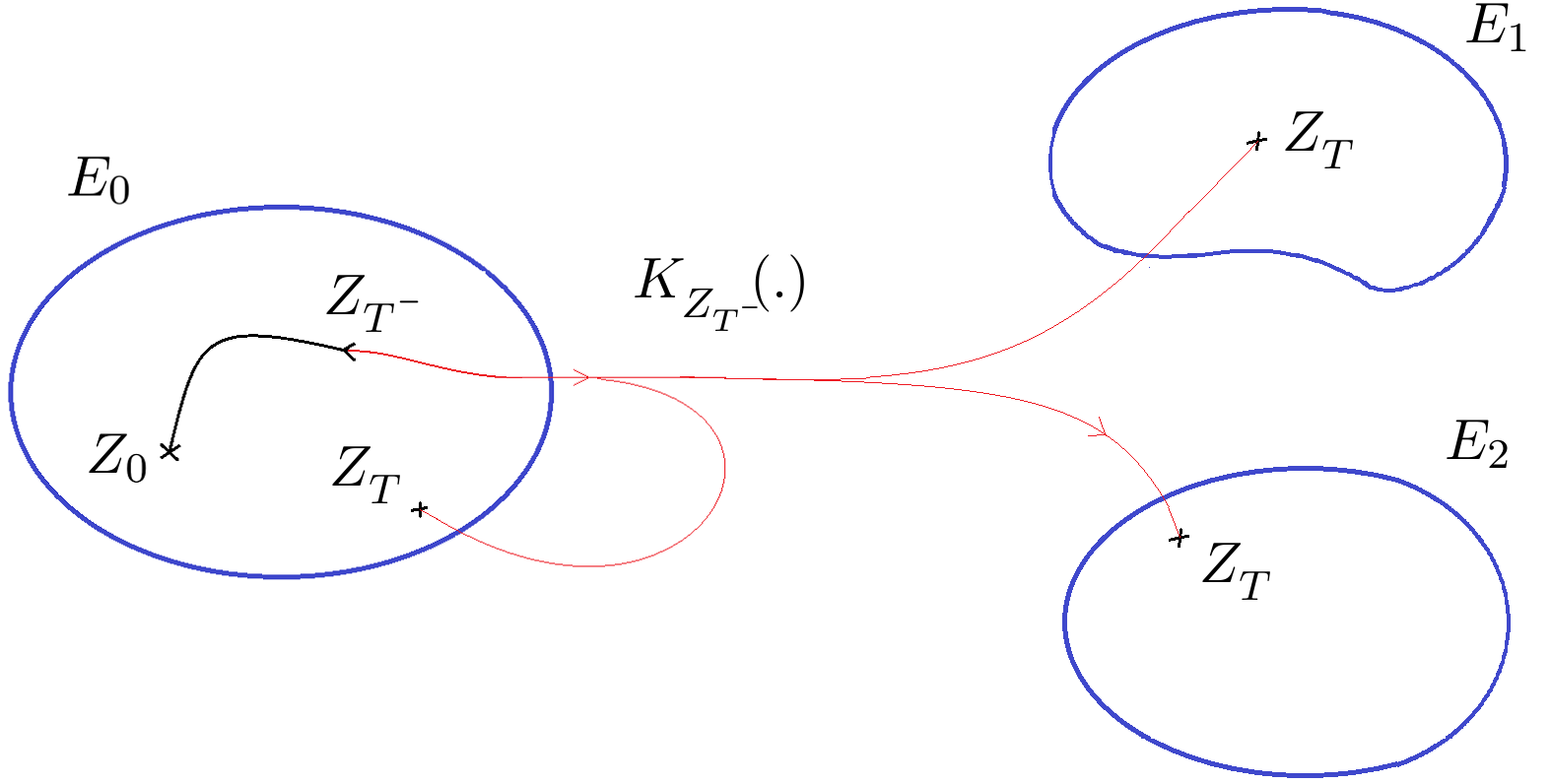}
 \caption{A spontaneous jump.\label{sautalea}} 
 \end{figure}
 The trajectory can also jump to another state when a random failure or a repair occurs (see Figure \ref{sautalea}). The distribution of the random time at which it happens is usually modeled through a state-related intensity function $\lambda:E\to\mathbb R_+$. For $z \in E$, $\lambda(z )$ represents the instantaneous probability (also called hazard rate) of having a failure or a repair at state $z $. If $Z_t=z$ and $T$ is the duration until the next jump, $\forall s< T$ we have $Z_{t+s}=\Phi_{z }(s)$. To simplify the notations in the following, we introduce the time-related intensity $\lambda_z$ such that $\lambda_{z}(s)=\lambda(\Phi_{z}(s))$ and $\Lambda_{z}(s)=\int_{0}^{s} \lambda \big(\Phi_{z}(u)\big) du$. 
If $\mathbb P_{z}(.)$ is the probability of an event knowing $Z_t=z$, we have: 
\begin{equation}
\mathbb P_{z}(T\leq s)=\left\{\begin{array}{cr}
1-\exp\left[-\Lambda_{z}(s)\right] & \mbox{ if }s<t^*_z\, ,\\
1& \mbox{ if }s\geq t^*_z\, .\\
\end{array}\right.
\label{interjump-time}
\end{equation} 
The law of $T$ has a continuous and a discrete part (see Figure \ref{probaT}). \begin{figure} [h]
\includegraphics[width=0.6\linewidth]{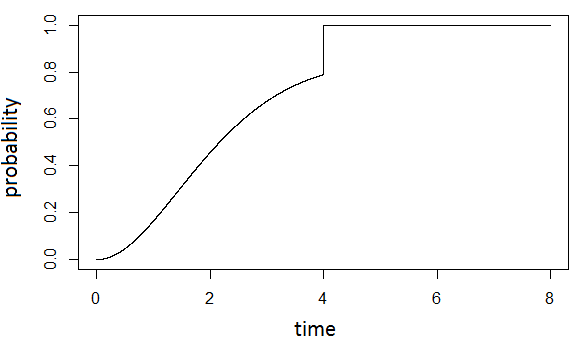}\centering
\caption{An example of the cumulative distribution function of $T$,\label{probaT}} 
{\small where $x\in\mathbb R^+$, $z=(x,m)\in E$, $\Phi_z(t)=(x+t,m)$, $\lambda(z)= \frac{x(5-x)}{12}$, and $t^*_z=4$ }
\end{figure}
As there is a discontinuity at $t_z^*$ in the cumulative distribution function of $T$, its reference measure  $T$ must include  a Dirac point at $t_z^*$ and therefore depends on $z$. We denote $\mu_z$ the reference measure of $T$  such that:
 \begin{align}
 \forall B\in \mathscr B\big([0,t_z^*]\big),\qquad \mu_{z}\big(B\big)=\mu_{Leb}\big(B\cup [0,t_z^*)\big)+\delta_{t_z^*}\big(B\big)  . \label{measureMu}
 \end{align} 
 This measure  will be useful to define the dominant measure $\zeta$ in Section \ref{sec:zeta}. It also allows to reformulate the law of $T$ with an integral form:
 \begin{align}
 \mathbb P_{z} (T\leq s) &=\int_{(0,s]} \bigg(\lambda_z(u) \bigg)^{\mathbbm{1}_{ u < t^*_z }}\exp\Big[-\Lambda_{z }(u)\Big] d\mu_z(u) \ .\label{Tlaw}
 \end{align}

 \subsection{A link between jump rate and the hazard rates of the possible transitions}
Note that the equation  \eqref{interjump-time}, or \eqref{Tlaw}, gives the time of the next jump, but it does not tell whether the transition is a failure, or a repair, or an automatic control mechanism. The type of the transition triggered is determined by the transition Kernel $\mathcal K_{Z_{t+T}^-}$.  For each jump,  we consider that there can be a countable number of possible transitions. Each type of transition is indexed by a number in the countable set $J$. When $Z_{t+T}^- \in \partial E$, $\mathcal K_{Z_{t+T}^-}$ can take an arbitrary form, but   when $Z_{t+T}^- \in  E$, the density of the kernel  $ K_{Z_{t+T}^-}$ is linked to the hazard rates of the possible transitions, as shown in the following. 
 If a transition is indexed by $j\in J$, we denote by $T^j$ the time between $t$ and the next occurrence of this transition, taking by convention $T^j=+\infty$ if the transition does not occur. This way the time of the next jump satisfies: \begin{equation}
    T=\min\left[\{T^{j},  \forall j\in J\}\cup \{t^*_z,\}\right]. \label{eq:TTj}
\end{equation}  Let $\lambda^j:E\to\mathbb R_+$ be its associated state-related intensity function, such that:\begin{equation}
 \forall s<t_z^*,\quad \mathbb P_{z}(T^j\leq s)=1-\exp\left[\text{-}\int_{0}^{s} \lambda^j \big(\Phi_z(u)\big)\right]du .\end{equation} 
 For instance, if the transition $j$ corresponds to a failure   the function $\lambda^j$ is the associated  failure rate,  and respectively if the transition $j$ corresponds to a repair, the function $\lambda^j$ is the associated  repair rate. Knowing  $Z_t=z_t=(x,m)$, and therefore, knowing the  path   given by $\phi_x^m$  that the positions are following, we make the assumption that the times $T^j$ are independent. This assumption is true if the position gathers all the variables affecting  the different types of transitions when the system is in mode $m$. According to the equation \eqref{eq:TTj}, this conditional independence implies that:
\begin{equation}
\forall z ^-  \in E,\qquad \lambda(z^-) =\sum_{j\in J} \lambda^j(z^-)\ . \label{lambda}
\end{equation}
Note the equations \eqref{lambda} is only valid when the   departure state $z^-$ is not on a boundary.
We denote by $B^j_{z^-}$  the possible arrival states of a jump when the transition $j$ is triggered and when the departure state is $z^-\in E$. We assume that the different types of transition are exclusive, meaning that for $i\neq j$ we have $B^i_{z^-}\cap B^j_{z^-}=\emptyset$. Then the probability of triggering the transition $i$ from the departure state $z^-$ is  $\mathcal K_{z^-}(B^i_{z^-})$, and we have:
\begin{equation}
 \forall z^-\in E,\qquad \mathcal K_{z^-}(B^i_{z^-}) = \frac{\lambda^i(z^-)}{\lambda(z^-)}\ . \label{eq:Klambda}
\end{equation}
Similarly the equation \eqref{eq:Klambda} is also valid only when the   departure state  $z^-$ is not on a boundary. When $z^-\in\partial E$ there is no link between $\lambda$ and $\mathcal K_{z^-}$.

\subsection{Generate a trajectory}
In order to generate a realization of the PDMP, one can follow these steps \cite{davis1984,davis1993markov,dufour2015numerical}: \begin{enumerate}
 \item Start at $t=0$ with state $Z_t=z_t$
 \item Generate $T$ the time until the next jump using \eqref{interjump-time} or  \eqref{Tlaw},  and \eqref{lambda} 
 \item Follow the flow $\Phi$ until $T$ using \eqref{flow}
 \item Generate $Z_{t+T}=z_{t+T}$ the arrival state of the jump knowing the departure state is $Z_{t+T^-} =\Phi_z(T)$ using \eqref{kernel} 
 \item Taking $t:=t+T$, repeat steps 1 to 5 until a trajectory of size $t_f$ is obtained
\end{enumerate}

 \subsection{Example}
 \label{subsec:sys3comp}
 As an example of a system, we consider a room heated by three identical heaters. $X_t$ represents the temperature of the room at time $t$. $x_e$ is the exterior temperature. $\beta_1$ is the rate of the heat transition with the exterior. $\beta_2$ is the heating power of each heater. The differential equation giving the evolution of the position (i.e. the temperature of the room) has the following form: $$\frac{d\,X_t}{dt}=\beta_1 (x_e-X_t)+\beta_2 \mathbbm 1_{M^1_t\, or\, M^2_t\, or\, M^3_t=ON}\ .$$ \begin{figure}[ht]\centering
 \includegraphics[width= \linewidth]{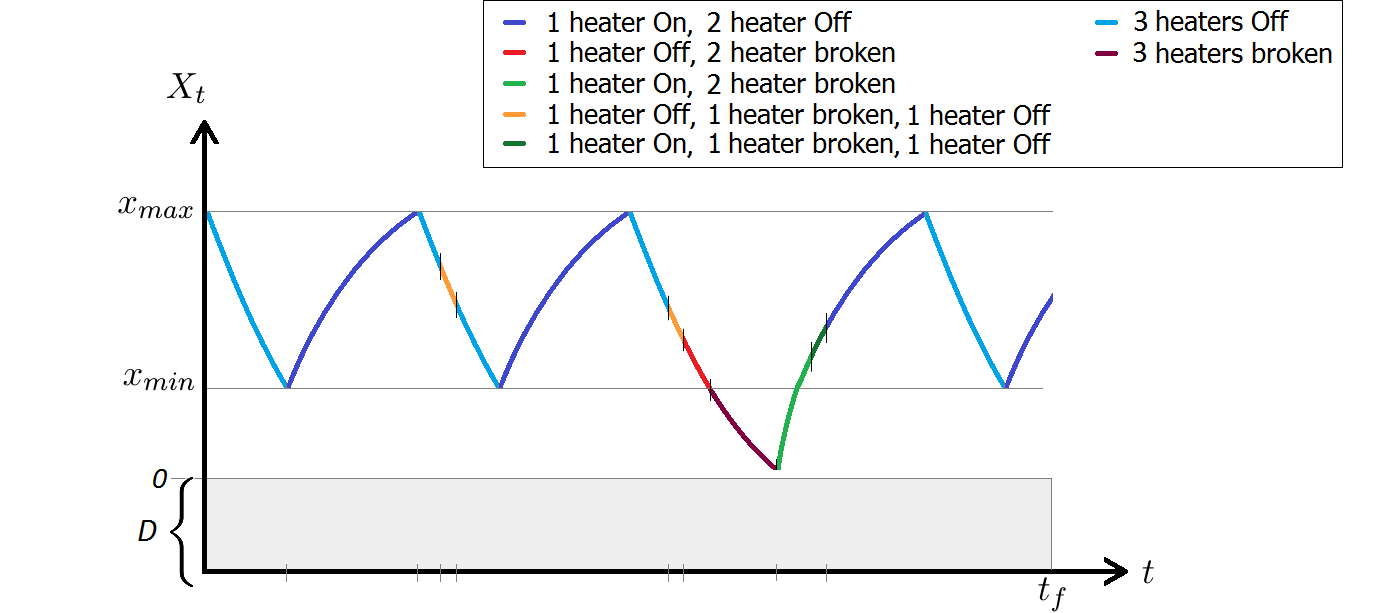} \caption{A possible trajectory of the heated-room system\label{fig:schemetraj}}\vspace{-0.7em}{\footnotesize(the mode is represented with colors)}
 \end{figure}
 
 The heaters are programmed to maintain the temperature within an interval\linebreak$(x_{min} , x_{max})$ where $x_e<0<x_{min}$. Heaters can be on, off, or out-of-order, so $\mathbb M=\{ON,OFF,F\}^{3}$. We consider that the three heaters are in passive redundancy in the sense that: when $X\leq x_{min}$ the second heater activates only if the first one is failed, and the third one activates only if the two other heaters are failed. When a repair of a heater occurs, if $X\leq x_{min}$ and all other heaters are failed the heater status is set to $ON$, else the heater status is set to $OFF$.
 To handle the programming of the heaters, we set \scalebox{0.9}{$\Omega_{m}=(-\infty,x_{max})$} when all the heaters are failed $m=(F,F,F)$ or when at least one is activated, otherwise we set \scalebox{0.9}{$\Omega_{m}=(x_{min},x_{max})$}. \\
 Due to the continuity of the temperature, the reference measure for the kernel is $\forall B\in\mathscr B (E),$ $\nu_{(x,m)}(B)=\sum_{m^+\in\mathbb M\backslash\{m\}}\delta_{(x,m^+)}(B)$.
 On the top boundary in $x_{max}$, heaters turn off with probability 1. On the bottom boundary in $x_{min}$, when a heater is supposed to turn on, there is a probability $\gamma=0.01$ that the heater will fail on demand. So, for instance, if $z^-=\big(x_{min}, (OFF,F,OFF)\big)$, we have $K_{z^-}\big(x_{min}, (ON,F,OFF)\big)=1-\gamma$, \linebreak and $K_{z^-}\big(x_{min}, (F,F,ON)\big)=\gamma(1-\gamma)$, and $K_{z^-}\big(x_{min}, (F,F,F)\big)=\gamma^2$.\\
 For the spontaneous jumps that happen outside boundaries, we consider the position is not modified during the jumps  and,   if the transition $j$   corresponds to the failure of a heater, then, for $z^-=(x^-,m^-m)\in E$, $\lambda^{j}(z^-)=0.0021+0.00015\times x^- $ 
and, if the transition $j$ corresponds to a repair, then, for $z^-=(x^-,m^-m)\in E$, $\lambda^{j}(z^-)=0.2$. A possible trajectory of the state of this system is depicted in figure \ref{fig:schemetraj}. Here the system failure occurs when the temperature of the room falls below zero, so $D=\{(x,m)\in E, x<0\}$.\\

\section{A reference measure for trajectories}
\label{sec:zeta}
We have seen in Section \ref{subsec:jumps} that when the position is restricted to a bounded set in some modes, the time to the next jump can be a hybrid random variable. We have to be cautious when considering the density of a trajectory of a PDMP for several reasons: first the reference measure for the times between the jumps is a mixture of Dirac and Lebesgue measures,
secondly these hybrid jumps may occur multiple times and in a nested way in the law of the trajectory of PDMP. Indeed, with these mixtures of Dirac and Lebesgue measures involved, the existence of a sigma-finite reference measure on the trajectory space is not obvious, yet it is mandatory to properly define the density of a trajectory. The existence of a reference measure is therefore crucial, because it preconditions the existence of the likelihood ratio needed to apply the importance sampling method.\\

We begin this Section by introducing a few notations: For a trajectory $\mathbf Z$ on the observation interval $[0,t_f)$, we denote by $N$ the number of jumps before $t_f$, and by $S_{k}$ the time of the $k$-th jump with the conventions $S_{0}=0$, and $S_{N+1}=t_f$. $\forall k<N ,\ T_k=S_{k+1}-S_{k}$ is the duration between two jumps and $T_N=S_{N+1}-S_N=t_f-S_N$ is the remaining duration between the last jump and $t_f$. One can easily verify that the sequence of the $(Z_{S_k},S_{k+1}-S_{k})$ is a Markov chain: it is called the embedded Markov chain of the PDMP \cite{davis1984}.\\ 

 \begin{figure}[h]\centering
 \includegraphics[width= 0.7\linewidth]{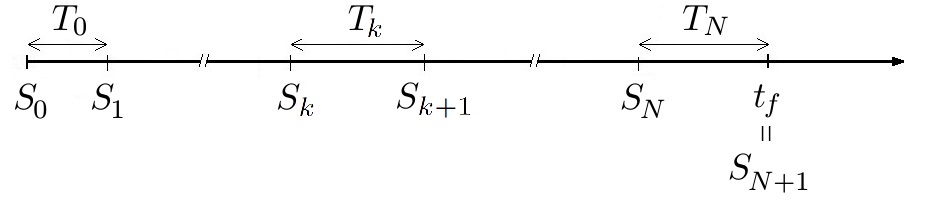}\vspace{-3ex}
 \caption{notations}\label{fig:notations}
 \end{figure} 

\subsection{The law of the trajectories}

The main idea in building the law of the trajectory $\mathbf{Z}$ is to summarize the trajectory by the truncated embedded Markov chain of the process: the vector $\big(Z_{S_0},T_0, \dots, Z_{S_N},T_N\big)$. This vector is also called the skeleton of the trajectory. As the trajectory is piecewise deterministic, we only need to keep the states of the arrivals of the jumps and the durations between the jumps to describe the trajectory. If we have the vector $\big(Z_{S_k},T_k\big)_{k\leq N}$ then we have enough information to reconstruct the trajectory using \eqref{flow} because we know the flow function $\Phi$. Noting $\Theta$ the map that changes $\mathbf{Z} $ into $\big(Z_{S_k},T_k\big)_{k\leq N}$, the law of $\mathbf{Z} $ can be defined as the image law of $\big(Z_{S_k},T_k\big)_{k\leq N}$ through $\Theta$. We denote by $\mathbf E$ the set of the trajectories defined on $[0,t_f)$. For $n\in\mathbb N$, let $A_n=\Big\{\big(z_{s_k},t_k\big)_{k\leq n}\in(E\times\mathbb R^{*}_\text{+})^{n},\, \overset{n}{\underset{i=0}{\sum}} t_i=t_f\Big\}$, so that $\Theta^{-1} (A_n)$ is the set of the trajectories including $n$ jumps. The sets $(\Theta^{-1} (A_n))_{n\in\mathbb N}$ form a partition of  $\mathbf E$. The sets $(A_n)_{n\in\mathbb N}$ form a partition of the set of the skeletons. \\

We can get the law of $\big(Z_{S_k},T_k\big)_{k\leq N}$, by using the dependencies between its coordinates. Thanks to \eqref{Tlaw} and \eqref{kernel} we can get the density of $T_k$ knowing $Z_{S_k}$ with respect to $\mu_{Z_{S_k}}$, and the density of $Z_{S_{k+1}}$ knowing $\big(Z_{S_k},T_k\big)$ with respect to $\nu_{Z_{S_{k+1}^-}} $, where $Z_{S_{k+1}^-}=\Phi_{Z_{S_k}}(T_k)$:
\begin{equation}
f_{T_k|Z_{S_k}=z}(u)=\Big(\lambda_{z}(u) \Big)^{\mathbbm{1}_{ u < t^*_{z} }}\exp\Big[\,\text{-}\,\Lambda_{z }(u)\Big]\ ,
\label{Tdensity}
\end{equation}
\begin{equation}
f_{Z_{S_{k+1}}|Z_{S_k},T_k}(z)=K_{Z_{S_{k+1}^-}}(z)\ .
\label{Z+density}
\end{equation} 
Using the Markov structure of the sequence $\big(Z_{S_k},T_k\big)_{k\leq N}$, the law of $\big(Z_{S_k},T_k\big)_{k\leq N}$ can be expressed as an integral of the product of the conditional densities given by \eqref{Tdensity} and \eqref{Z+density}. \\

We define the $\sigma$-algebra $\mathscr S$ on the set of the possible values of $\big(Z_{S_k},T_k\big)_{k\leq N}$ as the $\sigma$-algebra generated by the sets in 
 $\underset{\ n\in \mathbb N^* }{\bigcup} \mathscr B\Big(\Big\{\big(z_{s_k},t_k\big)_{k\leq n}\in(E\times\mathbb R^{*}_\text{+})^{n},\, \overset{n}{\underset{i=0}{\sum}} t_i=t_f\Big\}\Big)$.
 \begin{mydef}
 The law of the trajectory is then defined as follows, for $B\in \mathscr S$
 \begin{align}
 \mathbb P_{z_o}\Big(\mathbf{Z} \in \Theta^{-1}( B) \Big) = &\int_{B} \ \prod_{k=0}^{n} \Big(\lambda_{z_{k} }(t_k) \Big)^{\mathbbm{1}_{ t_k< t^*_{z_{k}}}}\exp\Big[-\Lambda_{z_{k}}(t_k)\Big] \prod_{k=1}^{n }K_{z_{k}^-}(z_{k})\nonumber\\
&\quad\times d\delta_{t^*_{n}}(t_n)\ d\nu_{z_{n }^-}(z_n)\ d\mu_{t^*_{z_{n-1}}}(t_{n-1}) \ ...\ d\nu_{z_{ 1}^-}(z_1)\ d\mu_{t^*_{z_{o}} }(t_{0})\ , 
\label{eq:loitraj}
 \end{align}
 where $z_j^-=\Phi_{z_{j-1}}(t_{j-1})$, and $t^*_{n}=t_f-\sum_{i=0}^{n-1}t_i$.
 \end{mydef}
 Note that, depending on the set $B$, $n$ can take different values in the equation \eqref{eq:loitraj}. Implicitly, the equation \eqref{eq:loitraj}, states that:  \begin{align}  \mathbb P_{z_o}\Big(\mathbf{Z} \in \Theta^{-1}( B) \Big)&=\mathbb P_{z_o}\bigg(\mathbf{Z} \in \Theta^{-1}\Big( \bigcup_{n\in \mathbb N}B\cap A_n\Big) \bigg)\nonumber\\
& =\sum_{n\in \mathbb N} \mathbb P_{z_o}\Big(\mathbf{Z} \in \Theta^{-1}( B\cap A_n) \Big)\nonumber\\
& =\sum_{n\in \mathbb N}\int_{B\cap A_n} \ \prod_{k=0}^{n} \Big(\lambda_{z_{k} }(t_k) \Big)^{\mathbbm{1}_{ t_k< t^*_{z_{k}}}}\exp\Big[-\Lambda_{z_{k}}(t_k)\Big] \prod_{k=1}^{n }K_{z_{k}^-}(z_{k})\nonumber\\
&\quad \qquad\times d\delta_{t^*_{n}}(t_n)\ d\nu_{z_{n }^-}(z_n)\ d\mu_{t^*_{z_{n-1}}}(t_{n-1}) \ ...\ d\nu_{z_{ 1}^-}(z_1)\ d\mu_{t^*_{z_{o}} }(t_{0})\ . \label{eq:loitraj2}
 \end{align} 
Also note that with our construction, this is a probability law on the space of the trajectories that satisfy \eqref{flow}, not on the set of all the trajectories with values in $E$.

 \subsection{The dominant measure and the density}
 \begin{mydef}
 We define the measure $\zeta$ so that 
 \begin{align}
 \zeta (\Theta^{-1}( B))= & \underset{ \mbox{\hspace{-12ex} } (z_{_k},t_{_k})_{k\leq n}
 \in B}{\int\quad d\delta_{t^*_{n}}(t_n)\ d\nu_{z_{n }^-}(z_n)}\ d\mu_{t^*_{z_{n-1}}}(t_{n-1}) \ ...\ d\nu_{z_{ 1}^-}(z_1)\ d\mu_{t^*_{z_{o}} }(t_{0}) \ .
 \label{zeta}
\end{align}
Note that, like in equation \eqref{eq:loitraj}, in the equation \eqref{zeta} $n$ can take different values depending on the set $B$.
\end{mydef}
\begin{myth}
If $\exists C>0, \forall z\in\xoverline{ E\,},\ \nu_z(E)<C$ and $t_f<\infty$, then $\zeta$ is a $\sigma$-finite measure. By Radon-Nikodym theorem, the density of a trajectory $\mathbf z=\Theta\big((z_{_0},t_{_0}),\, ...\, ,(z_{_n},t_n)\big)$ with respect to the measure $\zeta$ is 
 \begin{align}
f(\mathbf z) =\prod_{k=0}^{n} \Big(\lambda_{z_{k} }(t_k) \Big)^{^{\mbox{\hspace{-1.1ex}}\mathbbm{1}_{ t_k< t^*_{z_{k}}}}}\mbox{\hspace{-2.1ex}}\exp\Big[-\Lambda_{z_{k}}(t_k)\Big] \prod_{k=1}^{n }K_{z_{k}^-}(z_{k})\ ,
\label{density}
 \end{align}
 where $n$ is the number of jumps in the trajectory $\mathbf z$.
 \label{thm:sigmatefinite}
 \end{myth}
 The proof of Theorem \ref{thm:sigmatefinite} is given in appendix \ref{appendix-sec1}. \\ Note that it is always possible to choose the measures $\nu_{z^-}$ so they are all bounded by the same constant. Indeed the transition kernel is itself bounded by 1, as it is a probability measure. So, to get a measure $\zeta$ that is $\sigma$-finite, we can simply take the measures $\nu $
equal to the transition kernel, so the densities can be properly defined when the observation time $t_f$ is finite.

 \subsection{Admissible importance processes}
 Recall that an admissible importance process is any process whose law is absolutely continuous with respect to $\zeta$ (condition C2), and which has a density $g$ with respect to $\zeta$ satisfying $ \forall\, \mathbf{z} \in\mathscr D $, $\ f(\mathbf{z} )\ne 0 \Rightarrow g(\mathbf{z} )\ne 0$ (condition C3). In this Section, we  clarify the previous statement, and we identify to which extent we can modify the original process to obtain an admissible importance process. Throughout the rest of paper we denote the elements relative to this importance process with a $ ^\prime$, except for its density that is denoted by $g$.\\

 Our first remark is that condition C2 implies that the realizations of the importance process must satisfy equation \eqref{flow}. Indeed, the measure $\zeta$ involves the transformation $\Theta$ which uses the equation \eqref{flow} to rebuild a trajectory from a skeleton. Consequently, the importance process has to piecewisely follow the same flows as the original process. Similarly to the original process the importance process jumps to a new state for each change of flow. To ensure condition C2, the law of the $T^\prime_k$ has to be dominated by $\mu_{Z^\prime_{S^\prime_k}}$, and the law of $Z^\prime_{S^\prime_{k+1}}$ has to be dominated by $\nu_{Z^{\prime-}_{S_k} }$. This means that the boundaries of the $\Omega_m$'s and the set of the possible arrivals of a jump remain unchanged. So the modification of the original process focuses on the timing and nature of changes of modes, i.e. the laws of the jumps.\\
 
 To generate an importance process, we keep generating trajectories by successively generating the arrival state of a jump ($Z^\prime_{S^\prime_k}$) and the time until the next jump ($T^\prime_k $). As there is no requirement for the importance process to be Markovian, we consider that the law of a point of the trajectory $Z^\prime_t$ depends on the past values of states. As the states follow the flows piecewisely, it is equivalent to say that the law of $Z^\prime_{S^\prime_k}$ can depend on $\big(Z^\prime_{S^\prime_i},T^\prime_i\big)_{i< k}$, and that the law of $T^\prime_k$ can depend on $\big(Z^\prime_{S^\prime_i},T^\prime_i\big)_{i< k}$ and $Z^\prime_{S^\prime_k}$. 
For a jump time $S^\prime_k$, we denote $\underline Z^\prime_{S^\prime_k}=\big((Z^\prime_{S^\prime_i},T^\prime_i)_{i<k},Z^\prime_{S^\prime_k}\big)$, and we denote
 by $\lambda^\prime_{\underline z_k}(.)$ the intensity function associated to $T^\prime_k$ when $\underline Z^\prime_{S^\prime_k}=\underline z_k$. We have:
 \begin{align}
 \forall t\in (0,t^*_{z_k}],\quad \mathbb P(T_k^\prime\leq t|\underline Z^\prime_{S^\prime_k}=\underline z_k)&=\int_{(0,t]} \bigg(\lambda^\prime_{ \underline z_k}(u) \bigg)^{\mathbbm{1}_{ u < t^*_{z_k} }}\exp\Big[-\Lambda^\prime_{ \underline z_k}(u)\Big] d\mu_{z_k}(u) \label{integtime}
 \end{align}
 Noting ${\underline Z^\prime_{{S^\prime_k}^-}}=\big((Z^\prime_{S^\prime_i},T^\prime_i)_{i< k\,\text{-}1} \big)$ and $K^\prime_{\underline{z}^-}$ the importance kernel when $ \underline Z^\prime_{ {S^\prime_k}^- }={\underline{z}_{k} }^-$, we have:
 \begin{align}
 \forall B\in\mathscr{B}(E),\quad\mathbb P(Z^\prime_{S^\prime_k}\in B|{\underline Z^\prime_{{S^\prime_k}^-}}={\underline{z}_{k} }^-)
 &=\int_{B} K^\prime_{\underline{z}^-_k} (z) d\nu_{z^-_k}(z) \label{impkernel}
 \end{align}
Notice that the intensity function $\lambda^\prime_{\mathbf z_{s},s}$ in equation \eqref{integtime} does not have to be of the form $\lambda^\prime\circ\phi_{z_s}$, where $\lambda^\prime$ is a positive function on $E$. This means that at the time $S^\prime_k+t$, the   intensity does not depend only on the state~$Z^\prime_{S^\prime_k+t}$ as it would be the case if $\mathbf Z^\prime$ were a PDMP. So, in the importance process, we consider that the intensity can depend on the arrival state of the last jump and on previous pairs $(Z^\prime_{S^\prime_i},T^\prime_i)$. Therefore the importance process can be seen as a piecewise deterministic process (PDP) which is not necessarily Markovian. \\
 
 For condition C3 to be satisfied almost everywhere we can impose that almost everywhere for any $z_k\in E$, and $z_k^-\in\xoverline{ E\,}$, and $t\in (0,t^*_{z_k}]$ : \begin{align*}
 \mathbb E\big[\indic{\mathscr D}(\mathbf Z)\big| \underline Z_{S_k}=\underline z_k \big]>0,\mbox{ and } K_{z_k^-}(z_k)>0\ &\Rightarrow K^\prime_{\underline z_k^-}(z_k)>0\\
 \mathbb E\big[\indic{\mathscr D}(\mathbf Z)\big| \underline Z_{S_{k+1}^-}=(\underline z_k, t)\big]>0,\mbox{ and }\lambda_{z_k}(t)>0\ &\Rightarrow \lambda^\prime_{ \underline z_k}(t)>0.
 \end{align*} 
 Unfortunately with complex systems, the set $\mathscr D $ can be very hard to manipulate, and we do not always know if $\mathbb E\big[\indic{\mathscr D}(\mathbf Z)\big| \underline Z_{S_k}=\underline z_k \big]$ or $\mathbb E\big[\indic{\mathscr D}(\mathbf Z)\big| \underline Z_{S_k}=\underline z_k, T_k=t\big]$ are positive. So in practice we often only use the following sufficient condition which states that for almost any $z_k\in E$, and $z_k^-\in\xoverline{ E\,}$, and $t\in (0,t*_{z_k}]$ :
 \begin{align*}
 K_{z_k^-}(z_k)>0\ &\Rightarrow K^\prime_{\underline z_k^-}(z_k)>0\\
\lambda_{z_k}(t)>0\ &\Rightarrow \lambda^\prime_{ \underline z_k}(t)>0.
 \end{align*} \\

 \section{Optimal and practical importance process}
\label{sec:Opti}

\subsection{Practical importance processes and notations}
We will see in Subsection \ref{subsec:Opti} that we can restrict the search of an efficient importance process within a special class of processes without any loss in efficiency, because an optimal importance process (giving an estimator with zero variance) belongs to this special class.

The processes of this class are defined through the expressions \eqref{integtime} and \eqref{impkernel} but they do not use all the information contained in $\underline z_k$ and $\underline z_k^-$. The jump rates $\lambda^\prime_{ \underline z_k}(t)$ depend only on three variables which are : the current state $Z_{S_k+t}=\Phi_{z_k}(t)$, the time $t_f-(s_k+t)$ left before $t_f$, and the indicator $\indic{\tau_D\leq s_k+t}$ which tells if the system failure has already happened. The kernels $K^\prime_{\underline z_{k+1}^-}$ depend only on three variables, which are : the current departure state $z_{k+1}^-=\Phi_{z_k}(t_k)$, the time $t_f- s_{k+1} $ left before $t_f$, and the indicator $\indic{\tau_D\leq s_{k+1}}$.

So, to ease the presentation of such jump rates and transition kernels, we slightly modify the state space by adding an active boundary at the boundary of~$D$ and we add a coordinate on the mode which indicates if the trajectory has already visited $D$. The state now becomes $Z=\big(X,(M,M_D)\big)$ where $M_D= 0$ if $D$ has not been visited, and $1$ if it has. This way, for any time $t$ we have $Z_t=(X_t,(M_t,\indic{\tau_D\leq t}))$. For instance, with the heated-room system the set of modes becomes $\mathbb M=\{ON,OFF,F\}^3\times\{0,1\}$. The kernel $K_{Z^-}$ is unchanged when $M_D^-=M_D^+$, and is null when $M_D^-\ne M_D^+$, except at the boundary of $D$ where $K_{(0,( F,F,F ,0))}\big(0,( F,F,F,1)\big)=1$.

The three variables that determine the jump rates and kernels of the processes of the special class can now be identified by the current state and the current time. Therefore, we now consider importance processes with jump rate $\lambda^\prime_{ z_k,s_k}(t) $ and transition kernel $K^\prime_{ z^-_k, s_k} $. Such processes have the following laws of jump times and jump arrivals:
 \begin{align}
 \forall t\in (0,t^*_{z_k}],\quad &\mathbb P(T_k^\prime\leq t| Z^\prime_{S^\prime_k}= z_k, S_k^\prime=s_k)\nonumber\\
 &=\int_{(0,t]} \bigg(\lambda^\prime_{ z_k,s_k}(u) \bigg)^{\mathbbm{1}_{ u < t^*_{z_k} }}\exp\Big[-\Lambda^\prime_{ z_k,s_k}(u)\Big] d\mu_{z_k}(u) \label{integtime2}\\
 \forall B\in\mathscr{B}(E),\quad & \mathbb P(Z^\prime_{S^\prime_k}\in B|{ Z^\prime_{{S^\prime_k}^-}}={ {z}_{k} }^-, S_k^\prime=s_k)
 =\int_{B} K^\prime_{ z^-_k, s_k} (z) d\nu_{z^-_k}(z). \label{impkernel2}
 \end{align}
Note that the class of processes that can be defined by \eqref{integtime2} and \eqref{impkernel2} is included in the class of admissible importance processes. 
 Thanks to the new definition of the states, the conditional expectations \linebreak$\mathbb E\big[\indic{\mathscr D}(\mathbf Z)\big| \underline Z_{S_k}, T_k \geq t\big]$ are equal to the conditional expectations $\mathbb E\big[\indic{\mathscr D}(\mathbf Z)\big| Z_{S_k+t}=\Phi_{Z_{S_k}}(t) \big]$. This makes it possible to introduce the following important definitions:
\begin{mydef}
Let $U^*$ be the function defined on $E\times \mathbb R^+$ by:
 \begin{align}
 U^*(z,s)
 &=\mathbb E\big[\indic{\mathscr D}(\mathbf Z)| Z_{s}=z \big].
\end{align}
\end{mydef}
\begin{mydef}
Let $U^-$ be the function defined on $E\times \mathbb R^+$ by:
 \begin{align}
 U^{\text{-}}(z^{\text{-}},s) 
 &=\int_E U^{*}(z^+,s)K_{ z^{\text{-}}}(z^+)d\nu_{ z^{\text{-}}}(z^+).
\end{align}
\end{mydef}
 The quantity $U^*(z,s)$ measures the chances of having a system failure before $t_f$ knowing the system is in state $z$ at time $s$, and the quantity $U^-(z^-,s)$ the chances of having a system failure before $t_f$ knowing the system is jumping from the state $z^-$ at time $s$. These quantities play an important role in the latter. 
 
 \subsection{A way to build an optimal importance process}
 \label{subsec:Opti}
 In the importance process, generating the trajectories jump by jump by using \eqref{integtime2} and \eqref{impkernel2} is not restrictive in term of efficiency, as proved by the following theorem: 
 \begin{myth}
For all $z\in E$, $z^-\in\xoverline{E\,}$, and $s\in[0,t_f)$, the jump densities with respect to $\mu_z$ such that \begin{align}
 g^{*}_{T^\prime_k|Z^\prime_{S^\prime_k},\, S^\prime_k= z,s}(u)
 &=\frac{U^{\text{-}}\big(\Phi_{z}(u),s+u\big)}
 {U^{\text{*}}\big( z ,s\big)} f_{T_k| Z_{S_k}=z}(u)\ 
 \label{bestL}
\end{align}   and  the kernels $\mathcal K^{*}_{ { z}^{\text{-}},s} $ having a density with respect to $\nu_{z^-}$ which satisfies  \begin{align}
 K^{*}_{ { z}^{\text{-}},s}(z)&=\frac
 {U^*\big(z,s\big)}
 {U^{\text{-}}\big(z^{\text{-}},s\big)}K_{ z^{\text{-}}}(z)\ \label{bestK}
\end{align} correspond to the jump densities and the transition kernels of an optimal importance process. \\Note, these optimal densities do  integrate  to one as $U^{\text{*}}\big( z ,s\big)=\int_0^{t_z^*} U^{\text{-}}\big(\Phi_{z}(u),s+u\big)  
   f_{T_k| Z_{S_k}=z}(u) du$,  and \linebreak  $U^{-}\big( z ,s\big)=\int_E U^{*}(z^+,s)K_{ z}(z^+)d\nu_{ z}(z^+)$.
\end{myth}
\begin{proof}
 Assume the trajectory $\mathbf z=\Theta\big((z_{_0},t_{_0}),\, ...\, ,(z_{_n},t_n)\big)$ has been simulated with \eqref{bestL} and \eqref{bestK}. Then its density $g$ with respect to $\zeta$ is:\begin{align}
g(\mathbf{z} ) &=\prod_{k=0}^{n} g^{*}_{T^\prime_k|Z^\prime_{S^\prime_k},\, S^\prime_k= z_k,s_k}(t_k) \prod_{k=1}^{n }K^{*}_{ { z_k}^{\text{-}},s_k}(z_k) \nonumber\end{align}
So it verifies:
 \begin{align}
g(\mathbf{z} ) &=\prod_{k=0}^{n}\frac {U^{\text{-}}\big(\Phi_{z_k}(t_k),s_k+t_k\big)}
 {U^{\text{*}}\big( z_k ,s_k\big)}\prod_{k=1}^{n }\frac
 {U^*\big(z_k,s_k\big)}
 {U^{\text{-}}\big(z_k^{\text{-}},s_k\big)} \prod_{k=0}^{n} f_{T_k| Z_{S_k}=z_k}(t_k) \prod_{k=1}^{n } K_{ z_k^-}(z_k) \nonumber\\
 &=\prod_{k=0}^{n}\frac {U^{\text{-}}\big(z_{k+1}^{\text{-}},s_{k+1}\big)}
 {U^{\text{*}}\big( z_k ,s_k\big)}\prod_{k=0}^{n-1 }\frac{U^*\big(z_{k+1},s_{k+1}\big)}
 {U^{\text{-}}\big(z_{k+1}^{\text{-}},s_{k+1}\big)} f(\mathbf z) \nonumber\\
 &=\frac{U^{\text{-}}\big(z_{n+1}^{\text{-}},s_{n+1}\big)}
 {U^*\big(z_{0},s_{0}\big)} f(\mathbf z) =\frac{\indic{\mathscr D}(\mathbf z)f(\mathbf z)}{\mathbb E_{z_0}\big[\indic{\mathscr D}(\mathbf z)\big]}=g^{*}(\mathbf{z}),\nonumber\end{align}
 where $g^{*}(\mathbf{z})$ is the density for an estimator with zero variance. 
\end{proof}
Equations \eqref{bestL} and \eqref{bestK} serve as a guide to build an importance process: one should try to specify densities as close as possible to these equations so as to get an estimator variance as close as possible to the minimal zero variance.

\subsection{Observations on the optimal process}
As we do not know the explicit forms of $U^*$ and $U^-$, the construction of an importance process close to the optimal one is delicate. Nonetheless, the equations \eqref{bestL} and \eqref{bestK} can give us information on how to build an importance process in practice. In this Section, we investigate the properties of the optimal importance process and of the function $U^*$ with the aim of building a good and practical importance process.

For instance, we can get the expression of the jump rate of the optimal process. For the time of the $k$-th jump, by definition of the jump rate and knowing that $(Z^\prime_{S^\prime_k},\, S_k)= (z,s)$, we get :
\begin{align}
 \lambda^{*}_{ z,s}(u)&=\frac { g^{*}_{T^\prime_k|Z^\prime_{S^\prime_k},\, S^\prime_k= z,s}(u)}{1-\int_0^u g^{*}_{T^\prime_k|Z^\prime_{S^\prime_k},\, S^\prime_k= z,s}(v)dv}, \nonumber \\
 \Leftrightarrow \quad \lambda^{*}_{ z,s}(u)&=\frac {U^{\text{-}}\big(\Phi_{z}(u),s+u\big) \Big(\lambda_{ z}(u) \Big)^{\mathbbm{1}_{ u < t^*_{z} }}\exp\Big[-\Lambda_{ z}(u)\Big]}{\int_{(u,t^*_{z}]}U^{\text{-}}\big(\Phi_{z}(v),s+v\big)\Big(\lambda_{ z}(v) \Big)^{\mathbbm{1}_{v < t^*_{z} }}\exp\Big[-\Lambda_{ z}(v)\Big]d\mu_{z}(v)}\ .\label{LB1}
\end{align} Using some properties of $U^*$ and \eqref{LB1} we can prove the following theorem:
\begin{myth}
\label{th:optijumprate}
The jump rate of the optimal importance process defined by the densities \eqref{bestL} and \eqref{bestK} verifies:
 \begin{equation}
\lambda^{*}_{ z,s}(u)
 =\frac {U^{\text{-}}\big(\Phi_{z}(u),s+u\big) }{U^{*}\big(\Phi_{z}(u),s+u\big)} \lambda_{ z}(u) \ . \label{LB2}
\end{equation}\end{myth} 
The proof is provided in appendix \ref{appendix-sec2}. 

Note that the expression \eqref{LB2} can be easily interpreted. $\lambda^{*}_{ z,s}(u)$ corresponds to the jump rate at the state $Z_{s+u}=\Phi_{z}(u)$. $U^{*}\big(\Phi_{z}(u),s+u\big)$ is the probability of generating a failing trajectory if $Z_{s+u}=\Phi_{z}(u)$ and if there is no jump at time $s+u$. $U^-\big(\Phi_{z}(u),s+u\big)$ is the probability of generating a failing trajectory if there is a jump at time $s+u$ and if the departure state is $Z_{s+u^-}=\Phi_{z}(u)$. So the ratio $\dfrac {U^{\text{-}}\big(\Phi_{z}(u),s+u\big) }{U^{*}\big(\Phi_{z}(u),s+u\big)}$ is the factor multiplying the probability of generating a failing trajectory when there is a jump at time $s + u$.
The expression indicates that, in order to reach the zero variance, one should increase the original jump rate in the same proportion as a jump would increase the probability of getting a failing trajectory.

The Theorem \ref{th:optijumprate} is noteworthy, because in practice the law of the jump time is specified through the jump rate. Thus it can be used to specify the laws of the jump times of an importance process, as we will do in Section \ref{subsec:ISparam}.

 Also note that with equations \eqref{LB2} and \eqref{bestK} indicate that, once the region $D$ has been reached, the optimal process does not differ from the original process. Indeed if $\tau_D$ is the reaching time of the critical region $D$, then for $s\geq \tau_D$ we have for all states $ z$ and $z^- $, $U^*(z,s)=U^-(z^-,s)=1$ and so for $s\geq \tau_D$ we get $K^*_{z^-,s}=K_{z^-}$, and for $s+u\geq \tau_D$ we get $\lambda^*_{z,s}(u)=\lambda_{z}(u)$.

 As it plays an important role in the expression of the optimal process, we look for more information about the function $U^*$. We first notice that: if $\tau$ is a stopping time such that $t_f>\tau>s$, then \begin{align}U^*(z,s)&=\mathbb E\big[\indic{\mathscr D}(\mathbf Z)\big| Z_s=z\big]\nonumber\\
 &=\mathbb E\Big[ \mathbb E\big[\indic{\mathscr D}(\mathbf Z)\big| Z_{\tau} \big]\Big| Z_s=z\Big]\nonumber\\
 \mbox{and so}\quad U^*(z,s)&=\mathbb E\big[U^*(Z_{\tau},\tau)\big| Z_s=z\big].\label{eq:Ustop} \end{align}
 Using the equation \eqref{eq:Ustop} we can show the two following properties:
 \begin{myth}
 $U^*$ is kernel invariant on boundaries:
 \begin{equation} 
 \forall z\in E,\qquad U^{\text{-}}\big(\Phi_{z}(t_z^{*}),s+t_z^{*}\big)= \lim_{ t\nearrow t_z^{*} } U^{*}\big(\Phi_{z}(t),s+t\big)\ .
\label{closureCondition}
\end{equation}
\label{thm:closureCondition}
 \end{myth}
 \begin{myth}
If $u\to U^{\text{-}}\big(\Phi_{z}(u),s+u\big) $ and $u\to \lambda_{ z}(u)$ are continuous almost everywhere on $[0,t^*_z)$, then almost everywhere $U^*$ is differentiable along the flow, with:
\begin{equation}
\frac{\partial U^{*}\big(\Phi_{z}(v),s+v\big) }{\partial v}
= U^{*}\big(\Phi_{z}(v),s+v\big) \lambda_{ z}(v) -U^{\text{-}}\big(\Phi_{z}(v),s+v\big) \lambda_{ z}(v)
\label{Uderiv}
\end{equation}
\label{thm:Uderiv}
 \end{myth}
 The theorems \ref{thm:closureCondition} and \ref{thm:Uderiv} can in fact be seen as foreward Kolmogorov equations on $U^*$. A complete proof for these two properties is in the appendix \ref{appendix-sec2}. \\
 \subsection{A parametric importance process}
 \label{subsec:ISparam}
 In order to find an importance process that gives a good variance reduction, we usually restrict the search within a parametric family of importance densities. Then we rely on optimization routines to find the parameters yielding the best variance reduction. Here, we propose to use a parametric approximation of $U^*(z,s)$, and then combine it with equations \eqref{LB2} and \eqref{bestK} to get the form of the importance kernels and of the importance intensities. 
If we denote $U_\alpha(z,s)$ our approximation of $U^*(z,s)$, where  the parameter $\alpha$ belongs to the set $A_{param}$, and we set $U_\alpha^{\text{-}}\big(z^{\text{-}},s\big)=\int_E U_\alpha(w,s )K_{ z^{\text{-}}}(w)d\nu_{ z^{\text{-}}}(w)$, then the corresponding importance intensities and kernels are given by :
\begin{align}
 \lambda^{\prime}_{ z,s}(u)
 &=\frac {U^{\text{-}}_\alpha\big(\Phi_{z}(u),s+u\big) }{U_\alpha\big(\Phi_{z}(u),s+u\big)} \lambda_{ z}(u) \ ,
 \label{impL2}
\\[10pt]
 K^{\prime}_{ { z}^{\text{-}},s}(z^+)&=\frac
 {U_\alpha\big(z^+,s\big)}
 {U^{\text{-}}_\alpha\big(z^{\text{-}},s\big)}K_{ z^{\text{-}}}(z^+)\label{impK2}\ .
\end{align} \\
With these settings and notations, condition (C3) can be expressed as: \begin{align*}
 U^*( z_k,s_k)>0 , \mbox{ and } K_{z_k^-}(z_k)>0\ &\Rightarrow U_\alpha( z_k,s_k)>0 \\
 U^*( z_k,s_k+t)>0 ,\mbox{ and }\lambda_{z_k}(t)>0\ &\Rightarrow U_\alpha( z_k,s_k+t)>0 ,
 \end{align*} 
 for any $z_k\in E$, and $z_k^-\in\xoverline{ E\,}$, and $t\in (0,t*_{z_k}]$. It is therefore satisfied if we take $U_\alpha $ positive everywhere for instance.

 Here we switch the problem of setting a density $g$ close to $g^*$ by finding $\lambda^\prime$ and $K^\prime$, to the problem of finding a surface $U_\alpha$ on $E\times \mathbb R^+$ close to the surface $U^*$. 

Note that this way of building a parametric family of importance processes can be applied to any kind of systems, though the shape of $U_\alpha$ may have to be adapted from case to case. Indeed, we expect the shape of $U^*$ to depend on the configuration of the system and so does the shape of the $U_\alpha$'s.

We could also have plugged the approximations $U_\alpha$ and $U_\alpha^-$ into \eqref{bestL}, rather than into \eqref{LB2}, but the option we have chosen is in fact more convenient and computationally more efficient. With Equation \eqref{LB2}, we pass through the intensity, so the density of the $T_k^\prime$'s automatically integrates to 1. Conversely if we pass through equation \eqref{bestL}, we have to renormalize the density so it integrates to 1 before simulating a realization of the $T_k^\prime$. As this renormalization requires to compute an integral, it is less advantageous.

 \subsection{Remarks on the parameter optimization}
 As mentioned in the introduction, we propose to use the cross-entropy method presented in \cite{tutorialCE} to select the parameters of the importance density as it was done in \cite{zuliani2012rare}. In the case of PDMP, it is hard to use the adaptive cross-entropy method presented in \cite{tutorialCE}: the adaptive cross-entropy requires to have a function $\mathcal O:\mathbf E\to \mathbb R$ that orders the trajectories which is hard to specify judiciously. This function must order the states in such a way that there exists a threshold  $c$ for which \begin{equation}
 \mathbb P(\mathbf Z \in \mathscr D)=\mathbb P(\mathcal O(\mathbf Z)>c),\end{equation} and there exists a sequence of thresholds $ c_0\leq c_1\leq \dots \leq c_k=c$ so that \begin{equation}  \mathbb P(\mathbf Z \in \mathscr D_i)=\mathbb P(\mathcal O(\mathbf Z)>c_i),\end{equation} where $\mathscr D_0\subseteq\mathscr D_1\subseteq\dots\subseteq\mathscr D_i\subseteq \mathscr D_{i+1}\subseteq\dots\subseteq\mathscr D_{k}=\mathscr D$. In order to run the Cross-Entropy algorithm with relatively low sample sizes at each step (from 100 to 1000), it is good to set the function $\mathscr O$ so that   $20\leq\frac{\mathbb P(\mathcal O(\mathbf Z)>c_i)}{\mathbb P(\mathcal O(\mathbf Z)>c_{i+1})}\leq 100$ \cite{tutorialCE}. The issue is that  we find it hard to specify such a function  with PDMP.  For this reason we used a simplified version of the CE method considering only one threshold $c$.
 
The CE algorithm also requires to minimize an approximation of the Kullback-Leiber divergence  $D(g_\alpha,g^{*})$. We simulate a sample used to compute  many approximations of  $D(g_\alpha,g^{*})$. The  optimization routine uses the sample to compute some approximations of  $D(g_\alpha,g^{*})$ with different values of $\alpha$.  When the sample contains too many trajectories in $\mathscr D$, this approximations can be computationally heavy. Conversely, when the sample does not contain enough   trajectories in $\mathscr D$  the approximations are not accurate enough. So we choose to increase the size of the sample gradually until  it contains $n_{CE}$ trajectories in $\mathscr D$, $n_{CE}$ being a number fixed by the user. This way the objective function to minimize and its gradient are both a sum over $n_{CE}$ terms, and thus they are not too heavy to compute. The CE algorithm we used is presented in Table   \ref{tab:CE0}.
 \begin{table}[h]\fbox{
 \begin{algorithm}[H]
 \textbf{Initialization: } choose $\alpha_0\in A_{param}$ and $n_{CE}\in\mathbb N^*$ and set $ t=0$, and $ \varepsilon>0$\\
 \While{ $||\alpha_{t}-\alpha_{t+1}||<\varepsilon $}{
 Set $k=1$, and generate   $  \mathbf Z^\prime_1 \sim g_{\alpha_t}$\\
 \While{ $\sum_{i=1}^k \indic{\mathbf Z^\prime_i\in\mathscr D}<n_{CE}$}{Generate $\mathbf Z^\prime_{k+1} \sim g_{\alpha_t}$ \\ $k:=k+1$}
 $N=k-1$\\
 Compute $\alpha_{t+1}=\underset{\alpha\in A_{param}}{\mbox{argmin }}\frac 1 N \sum_{i=1}^{N} \indic{\mathbf Z^\prime_i\in\mathscr D}\frac{f(\mathbf Z^\prime_i)}{g_{\alpha_t}(\mathbf Z^\prime_i)} log\big(g_{\alpha_t}(\mathbf Z^\prime_i)\big) $ \\
 $t:=t+1$ }
 \textbf{End: }Estimate $p$ using the importance density $g_{\alpha_{t-1}}$ 
 	\end{algorithm}}
\caption{CE algorithm}
\label{tab:CE0}
\end{table}

 However, to our knowledge, there is no guarantee that the minimization routine used in the cross entropy method converges to a global optimum. Therefore, to avoid falling in a local optimum, one should run several times the cross entropy method with different initial values for the vector of parameters. Note that the parametrization must be chosen carefully: indeed the family of the importance densities must contain densities that are close to the zero-variance density $g^*(\mathbf z)=\frac{\indic{ \mathscr D }(\mathbf z)f(\mathbf z)}{p}$ to obtain a good variance reduction, otherwise we could even obtain a variance increase. In order to avoid   a variance increase, the parametric family should contain the original density $f$. Indeed, if we specify  in the parametric family that for say $\alpha=0$ we have $g_0=f$ then the parameter optimization  should not select  a  parameter worse than $\alpha=0$  and in the worst scenario the variance remains unchanged. This is why we advise that the family of the $U_\alpha$ functions includes a constant function, so that the original process with jump rate $\lambda_z$ and transition kernel $K_z$ is included in the admissible importance processes. \\

The initial vector of parameters $\alpha_0$ has a big influence on the convergence of the method. Ideally, it  should be chosen to simulate $n_{CE}$ trajectories in $\mathscr D$ relatively fast, but, in order to avoid an over-biasing situation with a wrong approximation of the Kullback-Leiber divergence at the first step, we recommend to choose $\alpha_0$ so that $g_{\alpha_0}$ is as close as possible from $f$. Testing several values of $\alpha_0$ is therefore necessary, to get a sense of what is a good $\alpha_0$.

\section{Simulation study on a test case}
\label{sec:Ex}

In this Section we present how we build an importance process for the heated room system presented in Section \ref{subsec:sys3comp}.\\ 

\subsection{A parametric family of importance processes}

 In the heated-room system, the three heaters are identical and are in parallel redundancy, so we expect the probability $U^*(z,s)=\mathbb E\big[\indic{\mathscr D}(\mathbf z)| Z_{s}=z\big]$ to increase with the number of failed heaters in the state $z$. Therefore, noting $b(z)$ the number of failed heaters in state $z$, we start by setting \begin{equation}
 U_\alpha(z,s)= H_\alpha\big( b(z)\big)\,Q(x,s)
 \end{equation} where $Q$ is a function of position and time, and $H_\alpha $ is a function on integers. We set $H_\alpha( 0)=1$. As we want $U_\alpha(z,s)$ to increase with $b(z)$, $H_\alpha $ has to be an increasing function.
 
 If $T$ denotes the time until the next jump after a time $s$, using \eqref{eq:Ustop} with $\tau=s+T$ we get: \begin{align} \quad U^*(z,s)&=\mathbb E\big[U^*(Z_{s+T},s+T)\big| Z_s=z\big]. \end{align} As the repair rates are larger than the failure rates by one order of magnitude in practice, when there is at least one failed heater, the probability of arriving in a more degraded state $Z_T$ is much lower than the probability of having a repair. This last remark can actually be applied to any reliable industrial system (see for instance \cite{Pycatshoo2}). Ideally we would like $U_\alpha$ to mimic the property of $ U^*$ so we would like to have 
 \begin{equation}
 U_\alpha (z,s)=\mathbb E\big[U_\alpha(Z_T,s+T)\big| Z_s=z\big] \label{mimic}
 \end{equation}
 which can be reformulated as :
 \begin{equation}
 H_\alpha \big(b(z)\big)=\sum_{m^+\in\mathbb{M}}H_\alpha\big(b(x,m^+)\big)\int_{(0,t_z^*]}\mbox{\hspace{-1.4em}}K_{\Phi_z(u)}\big((\phi_x^m(u),m^+)\big)w_z(u)d\mu_z(u) \label{mimic2}
 \end{equation}
 where $w_z(u)=\dfrac{Q (\phi_x^m(u),s+u)}{Q (x,s )}\exp\big[\text{-}\Lambda_z(u)\big]$.
 As a repair is much more likely than failure, if the transition from state $(\phi_x^m(u),m)$ to the state $(\phi_x^m(u),m^+)$ indexes a repair $K_{\Phi_z(u)}\big((\phi_x^m(u),m^+)\big)$ is larger than if it had indexed a failure. So, \eqref{mimic2} implies that, when $b(z)>1$, the value of $H_\alpha( b(z))$ is closer from $H_\alpha( b(z)-1)$ than from $H_\alpha( b(z)+1)$. As $H_\alpha$ was supposed increasing, it must be convex. So
 we propose that $H_\alpha( b(z) )=\exp\big[\alpha_1{b(z)}^2]$, with $\alpha_1>0$.    If, from a $Z_s=\Phi_z(u)$, the transition $j$ corresponds to a failure    then we have:
 \begin{equation}
 {\lambda^\prime}^{j}_{z,s}(u)= \lambda^{j} _{ z }(u)\exp\big[\alpha_1\big(2b(z)+1\big) ]\ , \label{lfimp}
 \end{equation}
 and if it corresponds to a repair then we have:
 \begin{equation}
   {\lambda^\prime}^{j}_{z,s} (u)= \lambda^{j}_{ z }(u)\exp\big[-\alpha_1\big(2b(z)-1\big)]\ . \label{lrimp}
 \end{equation}
 The jump rate satisfies:
 \begin{equation}
 {\lambda }^{\prime}_{ z,s}(u)= \sum_{i\in J}   {\lambda^\prime}^{i}_{z,s} (u)  \qquad \mbox{ and } \forall \qquad K^{\prime}_{ { \Phi_z(u)}^{\text{-}},s}(z^+)=\frac{ {\lambda^\prime}^{z^+}_{ z,s}(u)}{ \int_E {\lambda^\prime}^{z^+}_{z,s} (u)d\nu_z(z^+)}\ .
 \end{equation}
 We set the jump kernel such that its density satisfies for $u\in [0,t_z^*)$ :
  \begin{equation}
 K^\prime_{z^{\text{-}}}(z^+ )=\frac{K_{z^{\text{-}}}(z^+ )\exp \big[-\alpha_1\,b(z^+)^2\big]}{\int_E K_{z^{\text{-}}}(z)\exp\big[-\alpha_1\, b(z )^2\big]d\nu_{z^{\text{-}}}(z ) }.
 \end{equation}
Note that plugging $U_\alpha$ into the equations \eqref{bestL} and \eqref{bestK} imposes some kind of symmetry in the biasing of failure and repair rates. It is especially visible in equations \eqref{lfimp} and \eqref{lrimp}: On the one hand the failure rate associated to the transition from a state $z^-$ to $z^+$ is multiplied by a factor $\exp\big[\alpha_1\big(2b(z^-)+1\big) ]$, and on the other hand the repair rate corresponding to the reversed transition (from state $z^+$ to state $z^-$) is divided by a factor $\exp\big[\alpha_1\big(2b(z^-)-1\big) ]$. The equations \eqref{bestL} and \eqref{bestK} not only imply that the failures should be enhanced and the repairs inhibited, but it also states that the magnitudes of the distortion should be reciprocal.

 The square in $H_\alpha $'s formula was introduced to strengthen the failure rates when the number of broken heaters gets larger. The idea was to shorten the duration where several heaters are simultaneously failed in the simulated trajectories. Indeed, as repair is faster than failure, the shorter are the durations with a failed heater the more likely is the trajectory. Increasing the failure rates with the number of broken heaters is a mean to simulate more trajectories in $\mathscr D$ while maintaining the natural proportion between the likelihoods of the trajectories, which should decrease the variance.
 
 As the failure on demand was likely to play an important role in the system failure, we choose to separate it from spontaneous failure in our parametrisation setting $U_\alpha((x_{min},m),s)=\exp[-\alpha_2 b(z)^2]H_\alpha(x_{min},s)$. This allows to better fit $U_\alpha$ to $U^*$. Under this assumption, the equation \eqref{impK2} implies that for $z^-=(x_{min},m)$, the importance kernel takes this form:
\begin{equation}
 K^\prime_{z^{\text{-}}}(z^+ )=\frac{K_{z^{\text{-}}}(z^+ )\exp \big[-\alpha_2\,b(z^+)^2\big]}{\int_E K_{z^{\text{-}}}(z)\exp\big[-\alpha_2\, b(z )^2\big]d\nu_{z^{\text{-}}}(z ) }.
 \label{impK}
\end{equation}

\subsection{Results}
The Monte-Carlo simulations have been carried out using the Python library PyCATSHOO. (The flow functions $\phi^m_x$ were computed using a Runge-Kutta method of order 4 with a   discretization  step of 0.01. This discretization step is small enough so that reducing it further does not change the estimations.) As the Cross-Entropy method was not yet implemented in PyCATSHOO, we have used a specific Python code for the Cross-Entropy and the importance sampling methods. The system parameters used in the simulation were the following ones: $x_{min}=0.5$, $x_{max}=5.5$, $x_e=-1.5$, $\beta_1=0.1$, $\beta_2=5$, $t_f=100$. Trajectories were all initiated in the state $z_0=\big(7.5,\ (OFF,OFF,OFF)\big)$. The probability of having a system failure before $t_f$ was estimated to $p=1.29\times10^{-5}$ with an intensive Monte-Carlo estimation based on $10^8$ runs.
\renewcommand{\arraystretch}{1.5}
 \begin{table}[ht]\centering{ 
\begin{tabular}{c|c|c|c|c|c|c|@{}}
 \cline{2-7}
 & $N_{sim}$ & $\hat p$ & $\hat\sigma^2/N_{sim}$ & $ \widehat{\mbox{IC}}\times 10^5$ & $t_{sim}$ & $\widehat{eff}$\\
 \hline
\multicolumn{1}{|c|}{ \multirow{4}{*}{IS} } 
 & $10^3$ &$\, 1.28 \times10^{-5}$ & $\, 4.37\times10^{-13}$ & $\,[ 1.15 , 1.41 ] \,$& $0.073$ s & $3.1\times10^{10}$\\
\multicolumn{1}{|c|}{}& $10^4$ &$\, 1.273\times10^{-5}$ & $\, 5.07\times10^{-14}$ & $\,[ 1.228, 1.317 ] \,$& $0.073$ s & $2.7\times10^{10}$\\
\multicolumn{1}{|c|}{}& $10^5$ &$\, 1.289\times10^{-5}$ & $\, 5.01\times10^{-15}$ & $\,[ 1.275, 1.303 ] \,$& $0.077$ s & $2.6\times10^{10}$\\
\multicolumn{1}{|c|}{}& $10^6$ &$\, 1.288\times10^{-5}$ & $\, 5.05\times10^{-16}$ & $\,[ 1.283, 1.292 ] \,$& $0.079$ s & $2.5\times10^{10}$\\
 \hline
\multicolumn{1}{|c|}{\multirow{2}{*}{MC}} & $10^6$ & $\, 0.4 \times 10^{-5}$ & $\, 4.00\times10^{-12}$ & $\,[ 0.01, 0.79 ] \,$ & $0.022$ s & no convergence \\ 
\multicolumn{1}{|c|}{} & $10^7$ & $\, 1.3 \times 10^{-5}$ & $\, 1.28\times 10^{-12}$ & $\,[ 1.07,1.51 ] \,$ & $0.022$ s & $3.5 \times 10^6$\\ 
\hline
\end{tabular}}
\caption{Comparison between Monte-Carlo and importance sampling estimations}
\label{res}
 \end{table} 

 The values of the parameters selected by the cross-entropy method were $\alpha_1\simeq 0.915$ and $\alpha_2\simeq 1.197$, and for the first step, the approximation of the Kullback-Leiber divergence between $g^*$ and $g_\alpha $ was obtained by simulating from a biased density with parameters $(0.5,0.5)$. The whole cross-entropy method lasted approximately 9 minutes. Most of the running time was allocated to the optimization within each step of the cross-entropy, because each evaluation of the objective function and of its gradient was costly. In order to optimize the running time of the cross-entropy method, the size of the sample used for the approximations of the Kullback-Leiber divergence were set by simulating until we would get $n_{CE}=100$ trajectories with a system failure. The number of $n_{CE}=100$ roughly guaranties that the two first digit of of the Kullback-Leiber divergences are identified by their  approximations.   For each of the three steps needed to select the parameters,  samples of respectively 1970, 126, 127 trajectories were used.

 A comparison between Monte-Carlo and the associated importance sampling estimates is presented in Table \ref{res}, where we display the number $N_{sim}$ of simulations used for each method, the estimates $\hat p$ of the probability, the associated empirical variances $\hat\sigma^2/N_{sim}$ and confidence intervals $\widehat{\mbox{IC}}$, and the mean time of a simulation $t_{sim}$ in seconds. For $10^6$ simulations the results show that the Monte-Carlo estimator has not converged yet, whereas the importance sampling estimate is very accurate. To compare the two methods we estimate the efficiency of their estimators when they have converged. The efficiency is defined by the ratio of the precision and the computational time: $$eff=\frac{1}{\sigma^2/N_{sim}}\times \frac{1}{N_{sim} t_{sim}}=\frac{1}{\sigma^2t_{sim}}.$$ 
The efficiency can be interpreted as the contribution of a second of computation to the precision of the estimator. We estimate it by $\widehat{eff}= \frac{1}{\hat\sigma^2t_{sim}}$. The results indicate that our importance sampling strategy is approximately $7\,000$ times more efficient than a Monte-Carlo method. \\

We also verify that the importance sampling estimations are asymptotically normally distributed. The asymptotic normality was not observed for $N=10^3$, but it was observed for larger sample sizes. For instance for $N=10^4$,  the Figure \ref{fig:normallydistr} shows a normalized histogram on 100 estimations  $\hat p_{IS}$ that matches the normal density with mean $p$ and with the standard deviation of the 100 estimations.\begin{figure}[h]
    \centering
    \includegraphics[width=0.9\linewidth]{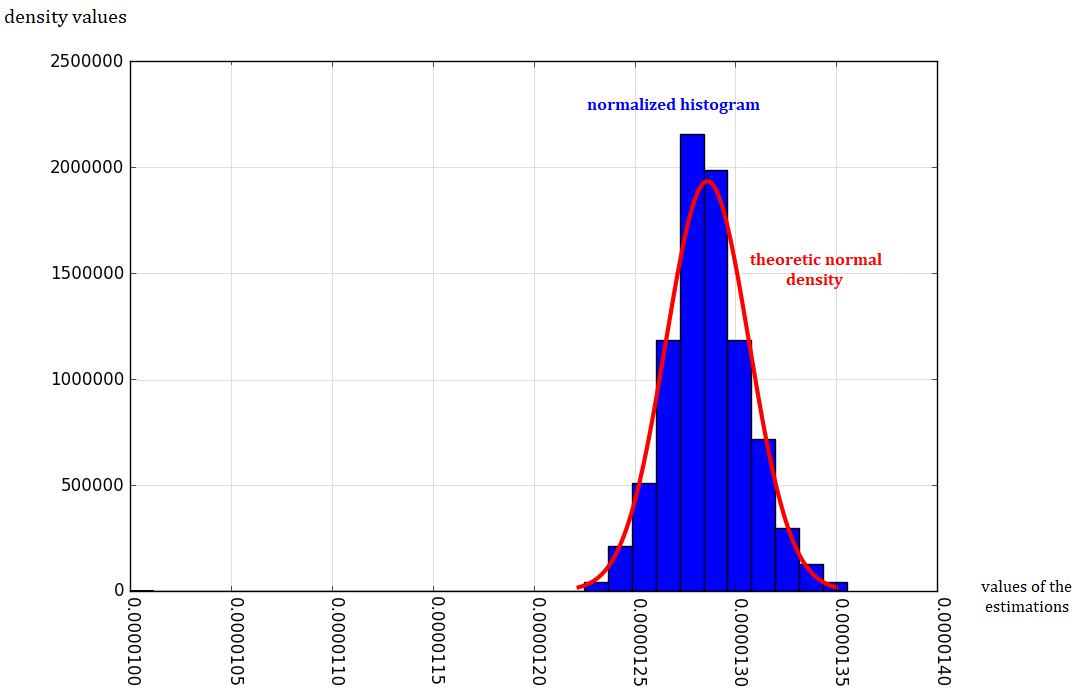}
    \caption{Asymptotic normality of the IS estimator (for $N=10^4$)}
    \label{fig:normallydistr}
\end{figure}
We also recorded the weights of the failing trajectories in the sample of one run of  the IS method with $N=10^4$. The Figure  \ref{fig:weights1}  shows that the weights are close to the value $p$, suggesting that the importance density is close to the optimal density. The figure \ref{fig:weights2} is a zoom-in on the largest weight: It shows there is no degenerated preponderant weight such that $\frac{f(\mathbf Z^\prime_i)}{g_\alpha(\mathbf Z^\prime_i)}\gg p$, suggesting there is no sign of under-favored region of $\mathscr D$ in $g_\alpha$. Here we do not need to check the weight degeneracy in all parts of $\mathcal D$ because, as we now the value of $p$ the can  simply check the estimation are unbiased and normally distributed to ensure convergence is reached. Finally, in Figures \ref{fig:traj1} and \ref{fig:traj2}, we present the graphs of two trajectories obtained respectively with the original process with density $f$ and with the importance process selected by the CE method with density $g_{(\alpha_1,\alpha_2)}$.
\begin{figure}[h]
    \centering
    \includegraphics[width=0.8\linewidth]{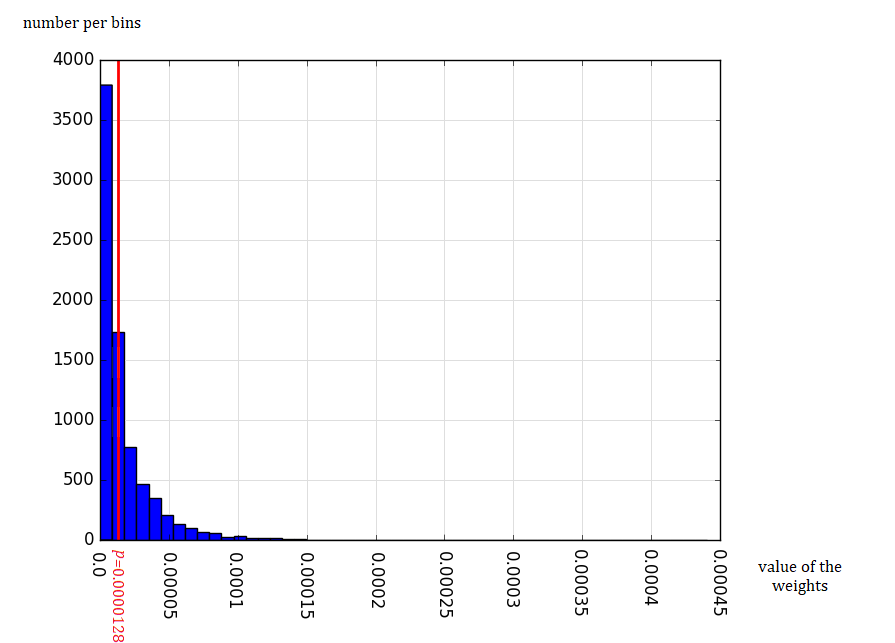}
    \caption{Allocation of the weights of failing trajectories (for $N=10^4$)}
    \label{fig:weights1}
\end{figure}
\begin{figure}[h]
    \includegraphics[width=0.8\linewidth]{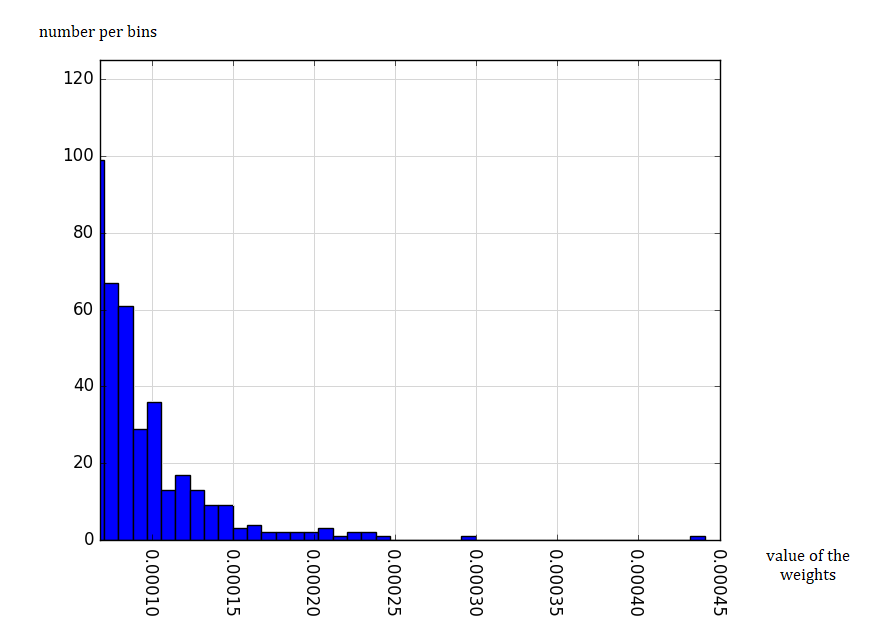}
    \caption{Allocation of the  largest weights in the sample (for $N=10^4$)}
    \label{fig:weights2}
\end{figure}

\begin{figure}
\centering
\begin{minipage}[c]{0.49\textwidth}
\centering
    \includegraphics[width=0.78\linewidth]{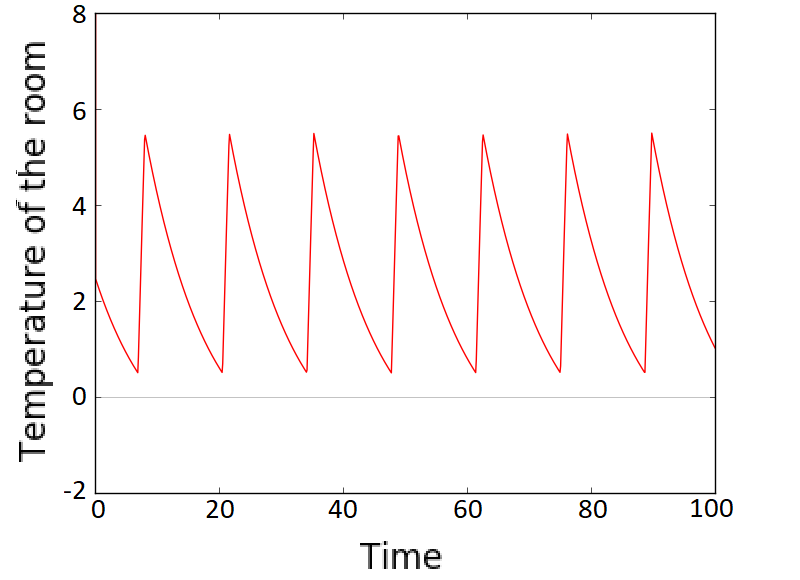}
    \includegraphics[width=0.78\linewidth]{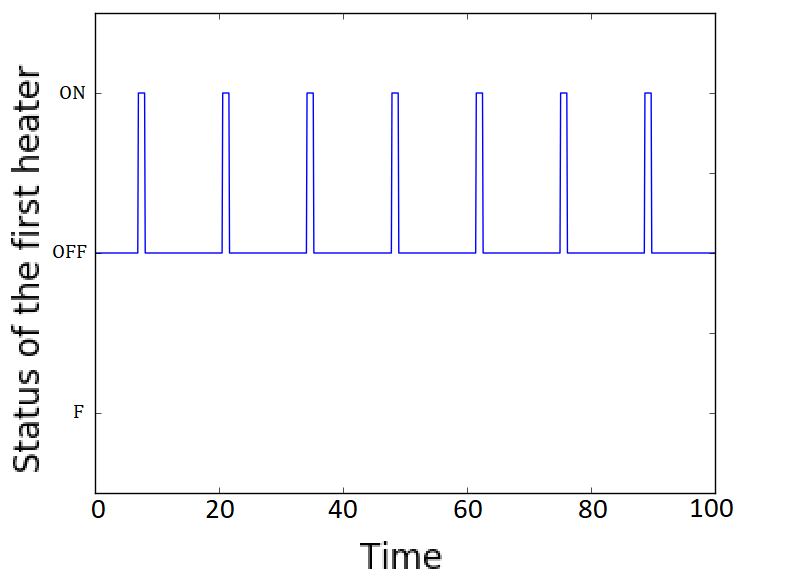}
    \includegraphics[width=0.78\linewidth]{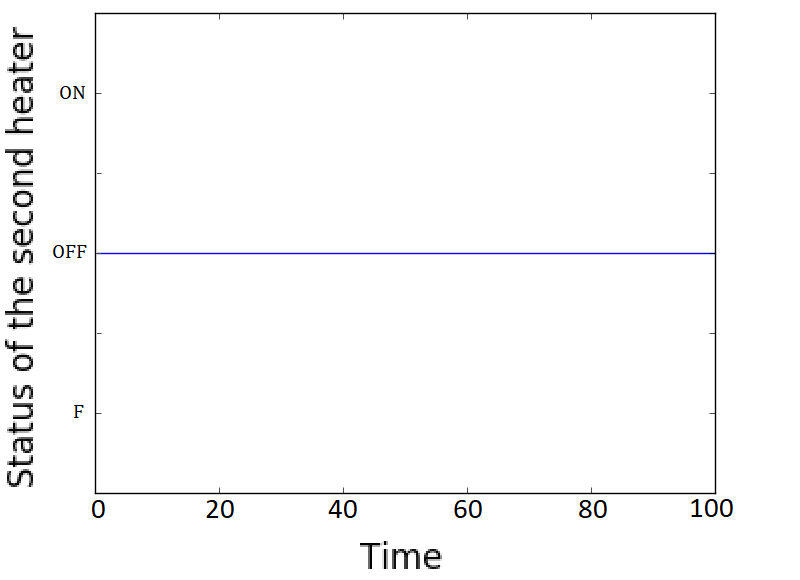}
    \includegraphics[width=0.78\linewidth]{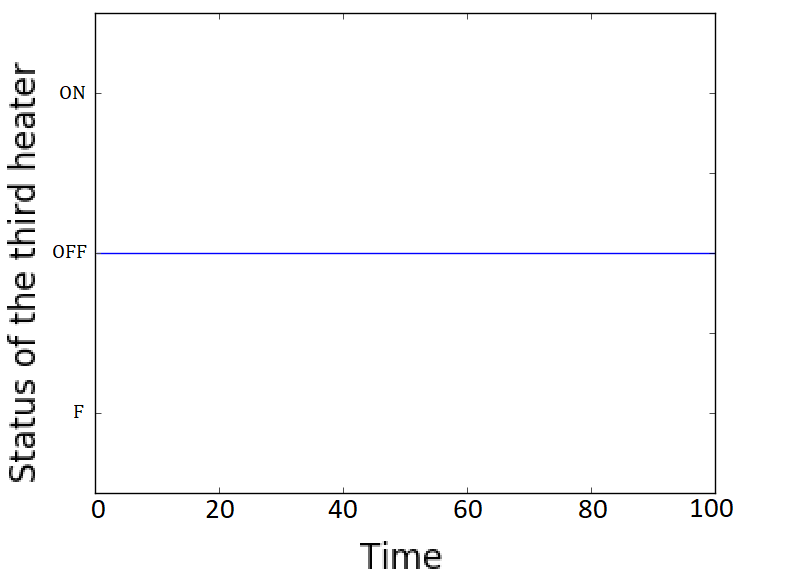}
    \caption{A trajectory of the   coordinates of the state of the system.  This trajectory was generated with the original process with density~$f$. }
    \label{fig:traj1}
\end{minipage}
\begin{minipage}[c]{0.49\textwidth}
\centering
    \includegraphics[width=0.78\linewidth]{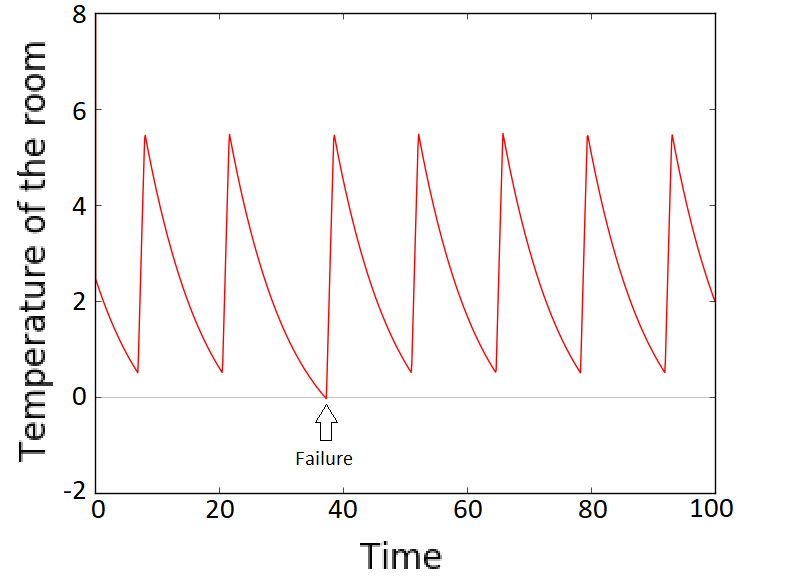}
    \includegraphics[width=0.78\linewidth]{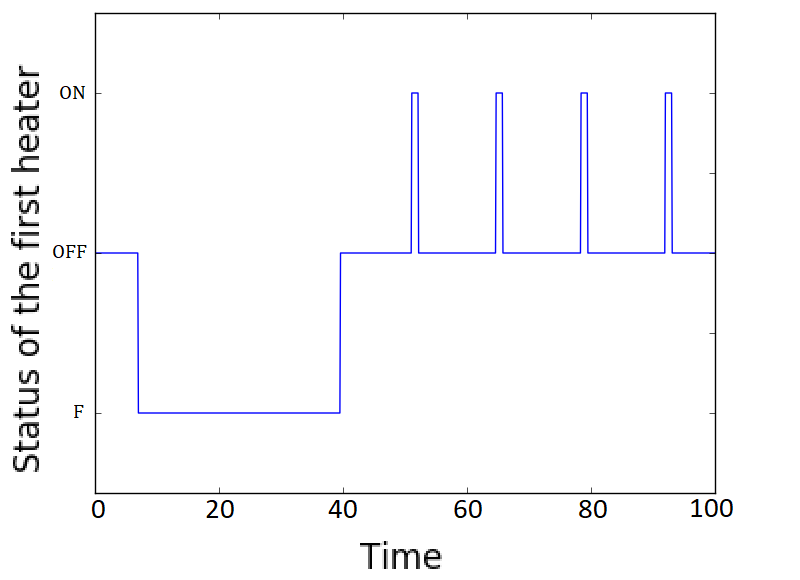}
    \includegraphics[width=0.78\linewidth]{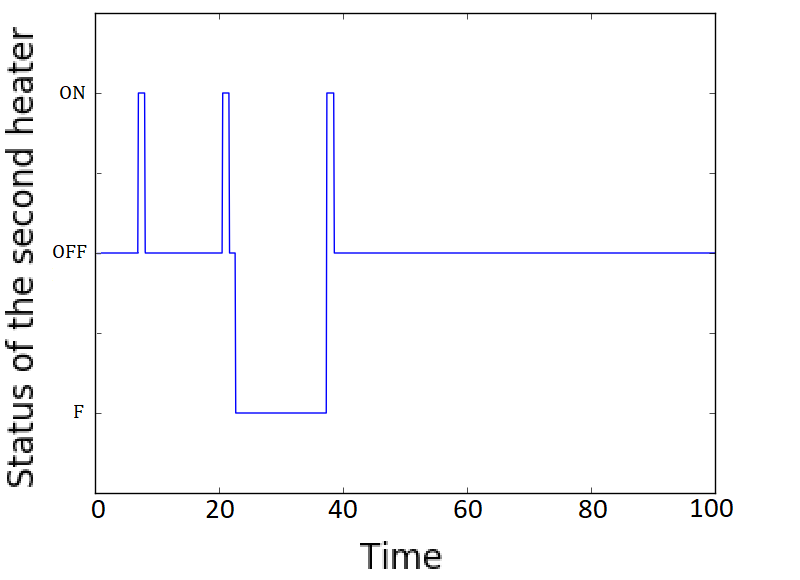}
    \includegraphics[width=0.78\linewidth]{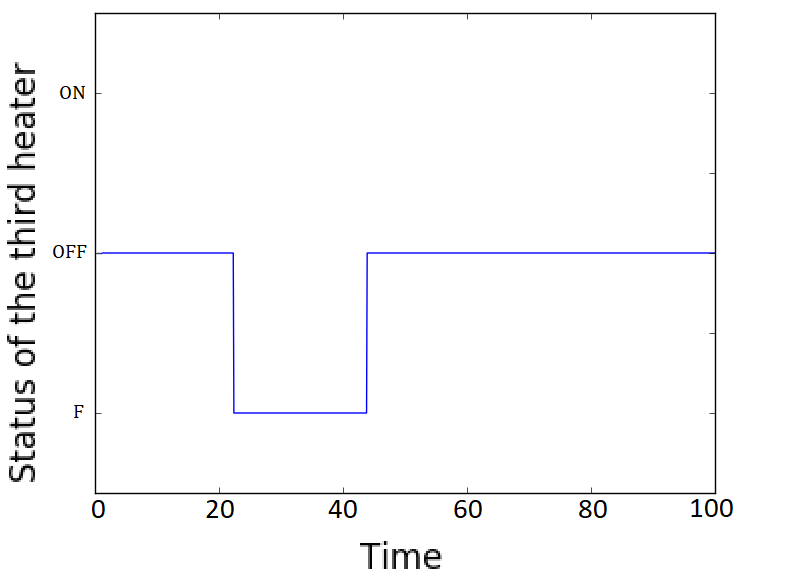}
    \caption{A trajectory of the   coordinates of the state of the system.  This trajectory was generated with the importance process with density $g_{(\alpha_1,\alpha_2)}$. }
    \label{fig:traj2}
\end{minipage}
\end{figure}

 \section{Discussion}
 Our work shows that importance sampling is applicable to any PDMP with or without boundaries. We have given the expressions of the intensities and the kernels of the optimal importance process, and we have seen that it depends on a critical function $U^*$. These expressions show that the optimal importance process has a specific structure. Although we do not have a closed form expression of the function $U^*$, these expressions are important for two reasons: 1)~They prove the existence of an optimal bias, which ensures that the importance sampling technique can be very efficient on PDMPs. 2)~They can guide the practical design of an efficient and explicit importance process. Indeed, by replacing $U^*$ by an approximation in the optimal expressions of the transition rates and kernels, we preserve the structure of the optimal importance process. The presented method therefore helps designing an importance process having the same behavior as the optimal one, and it showed good efficiency on our case study.
 
 This biasing strategy can be applied to any system, but the parametric shape of the approximation of $U^*$ may have to be adapted from case to case. The parametric shape presented in this article is suited to any system with similar components in terms of failure rates and repair rates and containing one minimal cut set (A minimal cut set being a group of components that need to fail so that the system can fail). For a system with a different configuration, we expect the shape of the function $U^*$ will differ, and the method may require a different parametric approximation for the function~$U^*$. 
 
 Our approach through the function $U^*$ can be applied to any sub-classes of PDMP, like, for instance, Markov chains \cite{kuruganti1996importance}, or continuous time Markov Chain, or queing models. In the particular case of PDMP that is a continuous time Markov Chain, the definition of the function $U^*$ is close to the forward committor function used in the transition path theory \cite{metzner2009transition}. In the case of a general PDMP, a committor function would be a function $(z,s)\to \mathbb E\big[\indic{\mathscr D_{A}}(\mathbf Z)| Z_{s}=z \big]$ where $\mathscr D_{A}$ is the set of trajectories that pass through $D$ without passing through a set $A\subset E$ first. $ U^*$ is therefore a commitor function for which $A=\emptyset$. It is also interesting to note that, in the Adaptive Multilevel Splitting algorithm, the asymptotic variance is minimized when using the committor function as the score function \cite{brehier2015analysis,cerou2019asymptotic}, similarly, in the interacting particles system method\cite{del2005genealogical}, the function $U^*$ also plays a role in the optimal potential function    method\cite{chraibi2018optimal}. Approaching the function $U^*$, allows to efficiently estimate rare event for importance sampling, but also for Adaptive Multilevel Splitting algorithm and the interacting particles system method. A method that allows to approximate this function would lead to significant improvement in the reliability assessment field.

 We proposed to find a good approximation of $U^*$ searching inside a family of parametric functions $(U_\alpha)_{\alpha\in A_{param}}$. In our application on the heated room system, we used the Cross-Entropy method to select an efficient parameter $\alpha$.
  We have noticed that the Cross-Entropy method tends to diverge quickly if it is not well initialized: The choices  of  $\alpha_0$ and $n_{CE}$ are critical for the convergence of the method. These two parameters must be well tuned because they impact the quality of the first approximations of the Kullback-Leiber divergence within the CE algorithm, and these approximations must be accurate enough to launch the optimization routine on a good track.  To choose a high value for $n_{CE}$ is a way to insure that these first approximations are accurate enough, but it is not worth considering in practice, as it greatly slows down the CE algorithm. The only solution is to find right away  an $\alpha_0$ which yields  correct approximations of the Kullback-Leiber divergence.  This is unfortunately  difficult to do, and may be even harder with more complex systems. We believe that the CE method used in this article must be improved or substituted by an other parameter optimization method so that the initialization gets less critical.

   With the IS method, depending on the importance process chosen, we can observe some weight degeneracy and therefore slow convergence. Weight degeneracy typically happens when two conditions are met: 1) there is a  domain $\mathcal D_1 \subset\mathcal D$ such that  the likelihood ratios within this domain are very big compared to the likelihood ratios in other domains of $\mathcal D$,  meaning the domain $\mathcal D_1$ is under-favored by the importance density $g$ compared to $f$; and 2) though it is unlikely,   few realizations of the importance process $\mathbf{Z}_i^\prime$ are drawn in $\mathcal D_1$, which creates unbalanced weights. 
Some methods allow to reduce the risk of weight degeneracy by using resampling schemes like for instance in the interacting particle system (IPS) method \cite{del2005genealogical}, but, even though it is reduced, the risk of weight degeneracy still remains within  the IPS method. The IPS method takes in input some potential functions $G_k$ (also called score functions). If these functions does not favor the domain $\mathcal D_1$, the convergence is slowed down\cite{chraibi2018optimal},  and we can end up with the same situation. In this method  the weights of the simulation outputs are the  inverse of the product of  resampling weights multiplied by the objective function's evaluations (see equations 2.18 in \cite{del2005genealogical}).  The  degenerate weights would therefore  appear each time   some trajectories in $\mathcal D_1$ are selected by the resamplings, eventhough the resamplings make it unlikely. Weight degeneracy does not depend on the method used but it rather depend on the choice of the importance process or on the choice the potential functions. Weight degeneracy is that it is  the symptom of a slow convergence, therefore the sample size should be increased until the weight degeneracy fades out: the weight degeneracy is a tool to select  the sample size for both methods. For a fixed sample size, if the sample contains some trajectories in $\mathcal D_1$, the weight degeneracy can be a criterion to reject the importance density, or the potential function, used. But this last criterion is valid only if the sample contains observations in the under-favored domains in $\mathcal D$, which is unlikely by definition. One important point to stress out, is that witnessing no weight degeneracy within the simulations outputs does not guarantee the convergence, we can consider we have converge if the sample size is reasonably large and that  we do not witness weight degeneracy in all part of $\mathcal D$.

 When choosing the importance process, there is a risk of over-biasing. Over-biasing  corresponds to  the situation where a domain    $\mathcal D_2\subset\mathcal D$ is over-favored by the importance process, resulting in an under-favoring of an other domains $\mathcal D_1\subset\mathcal D$. In this situation a weight degeneracy exist in   $\mathcal D_1$ but it is not witnessed  because no trajectory within the sample is drawn in $\mathcal D_1$. This situation happens when one type of failing trajectories is over represented in the importance distribution comparatively to other types of failing trajectories. This phenomenon can result in underestimating the probability of the system failure and in underestimating the variance. To avoid it, we must satisfy two points: 1) We must design a parametric importance density that can increase the likelihoods of each type of failing trajectories separately. 2) We need to initiate the Cross-Entropy method with a sample of trajectories that contains all types of failing trajectories. It is therefore preferable to apply this method only on systems of reasonable complexity, for which it is possible to determine the different types of failing trajectories.  The   parametric functions $(U_\alpha)_{\alpha\in A_{param}}$  should be flexible enough to satisfy the two previous points, but one should  pay attention to keep the dimension of the vector of parameter $\alpha$ reasonably small,  so that we avoid  a  prohibitive computational effort during the optimization routines in the CE.\\

 \section{Conclusion}
 We have presented a model for multi-component systems based on PDMPs. In order to speed up reliability assessment on such systems, we have adapted the importance sampling method to trajectories of PDMP. We have given a dominant measure for PDMP trajectories, allowing to properly define the likelihood ratio needed to apply the importance sampling method on such processes. The possible kinds of importance processes were discussed, and the optimal biasing strategy when simulating jump by jump was exhibited. We developed and tested a biasing strategy for a three-component heated-room system. Our importance sampling method has shown good performance, increasing the efficiency of the estimator by a factor $7\,000$.
 \newpage
 
\appendix
\section{The measure $\zeta$ is $\mathbf \sigma$-finite\\ when $\mathbf{t_f<\infty}$ and the measures $\mathbf{\nu_{z^-}}$ are bounded}
\label{appendix-sec1}
Remember that we defined the $\sigma$-algebra $\mathscr S$ on the set of the possible values of $\big(Z_{S_k},T_k\big)_{k\leq N}$ as the $\sigma$-algebra generated by the sets in 
 $\underset{\ n\in \mathbb N^* }{\bigcup} \mathscr B\Big(\Big\{\big(z_{s_k},t_k\big)_{k\leq n}\in(E\times\mathbb R^{*}_\text{+})^{n},\, \overset{n}{\underset{i=0}{\sum}} t_i=t_f\Big\}\Big)$. The measure $\zeta$ is defined by: : 
 \begin{align}
 B\in \mathscr S,\quad \zeta \big(\Theta^{-1}(B ) \big)= & \underset{ \mbox{\hspace{-4ex} } (z_{_k},t_{_k})_{k\leq n}
 \in B}{\int\qquad d\delta_{t^*_{n}}(t_n)}\ d\nu_{z_{n }^-}(z_n)\ d\mu_{t^*_{z_{n-1}}}(t_{n-1}) \ ...\ d\nu_{z_{ 1}^-}(z_1)\ d\mu_{t^*_{z_{o}} }(t_{0}) 
 \label{zetaApp}
 \end{align}
 \begin{proof}
Let $A_n=\Big\{\big(z_{s_k},t_k\big)_{k\leq n}\in(E\times\mathbb R^{*}_\text{+})^{n},\, \overset{n}{\underset{i=0}{\sum}} t_i=t_f\Big\}$. Then $\Theta^{-1}(A_n)$ is the set of possible trajectories with $n$ jumps, and the sets $A_n$ for $n\in\mathbb N^*$ form a partition of the set of all possible trajectories. Note that
$A_n\subseteq (E\times[0,t_f))^n$, so 
\begin{align*}
 \zeta \big(\Theta^{-1}(A_n) \big)&\leq \zeta (\Theta^{-1}\big( (E\times[0,t_f))^n \big)\\
 &\leq\underset{ \mbox{\hspace{-6ex}} (E\times[0,t_f))^n}{\int\qquad d\delta_{t^*_{n}}(t_n)}\ d\nu_{z_{n }^-}(z_n) \ d\mu_{t^*_{z_{n-1}}}(t_{n-1}) \ ...\ d\nu_{z_{ 1}^-}(z_1)\ d\mu_{t^*_{z_{o}} }(t_{0})
\end{align*}
We suppose that the $\nu_{z^-}$ are bounded, $\exists M>0, \forall z^-\in\xoverline{E\,},\ \nu_{z^-}(E)<M$. Under this assumption, we have:
\begin{align*}
\zeta \big(\Theta^{-1}(A_n) \big)&\leq M \underset{ \mbox{\hspace{-8ex} }(E\times[0,t_f))^{n-1}}{\int\qquad d\mu_{t^*_{z_{n-1}}}(t_{n-1})} \ ...\ d\nu_{z_{ 1}^-}(z_1)\ d\mu_{t^*_{z_{o}} }(t_{0}) \\
&\leq M \underset{ (E\times[0,t_f))^{n-2}}{\int\qquad } \int_E\int_{[0,t_f)} d\mu_{t^*_{z_{n-1}}}(t_{n-1})\ d\nu_{z_{n- 1}^-}(z_{n-1}) \ ...\ d\nu_{z_{ 1}^-}(z_1)\ d\mu_{t^*_{z_{o}} }(t_{0})\\
&\leq M(t_f+1) \underset{ (E\times[0,t_f))^{n-2}}{\int\qquad } \int_E d\nu_{z_{n- 1}^-}(z_{n-1})d\mu_{t^*_{z_{n-2}}}(t_{n-2}) \ ...\ d\nu_{z_{ 1}^-}(z_1)\ d\mu_{t^*_{z_{o}} }(t_{0})\\
&\leq M^2(t_f+1) \underset{ (E\times[0,t_f))^{n-2}}{\int\qquad } d\mu_{t^*_{z_{n-2}}}(t_{n-2})d\nu_{z_{n- 2}^-}(z_{n-2})\ ...\ d\nu_{z_{ 1}^-}(z_1)\ d\mu_{t^*_{z_{o}} }(t_{0}) 
\end{align*}
By recurrence we get that $\zeta \big(\Theta^{-1}(A_n) \big)\leq M^n(t_f+1)^n$, which proves that $\zeta$ is $\sigma$-finite.\end{proof}

\newpage

\section{Optimal intensity's expression, and some properties of $U^*$}
\label{appendix-sec2}
 \subsection{Equality \eqref{closureCondition} }
 \label{appendix-sec21}
 
 Let $z^\text{-}\in \delta E $ and $s\in[0,t_f)$. Remember that equality \eqref{closureCondition} states that 
$$
 U^{\text{-}}\big(\Phi_{z}(t_z^{*}),s+t_z^{*}\big)= \lim_{ 
 t\nearrow t_z^{*} } U^{*}\big(\Phi_{z}(t),s+t\big).$$
 \begin{proof}
 We denote by $T$ the time until the next jump after the trajectory has reached the state $Z_{s+t}=\phi_{z}(t)$. Then we have:
 \begin{align*}
 U^{*}\big(\Phi_{z}(t),s+t\big)&=\mathbb E\big[\indic{\mathscr D}(\mathbf z)\big|Z_{s+t}=\phi_{z}(t)\big]\\
 &=\mathbb E\Big[\mathbb E\big[\indic{\mathscr D}(\mathbf z)\big|Z_{T+s+t}\big]\Big|Z_{s+t}=\phi_{z}(t)\Big]\\
 &=\mathbb E\Big[(\indic{T<t^*_{\Phi_z(t)}}+\indic{T=t^*_{\Phi_z(t)}})U^*(Z_{T+s+t},s+t+T) \Big|Z_{s+t}=\phi_{z}(t)\Big]\\
 &= \int_0^{t^*_{\Phi_z(t)}} U^-( \Phi_{\Phi_z(t)}(u),s+t+u)\lambda_{\Phi_z(t)}( u)\exp\big[-\Lambda_{\Phi_z(t)}(u)\big]du \\
 &\quad + \exp\big[-\Lambda_{\Phi_z(t)}( t^*_{\Phi_z(t)})\big]\int_E K_{z^-}(z^+)U^*(z^+,s+t+t^*_{\Phi_z(t)})d\nu_{z^-}(z^+) \\
 &\mbox{where } z^-=\Phi_{\Phi_z(t)}(t^*_{\Phi_z(t)})\\
U^{*}\big(\Phi_{z}(t),s+t\big)
 &= \int_t^{t^*_{z}} U^-( \Phi_{z}(u),s+ u)\lambda_{z}( u)\exp\big[-\Lambda_{\Phi_z(t)}(u-t)\big] du \\
 &\quad + \exp\big[-\Lambda_{\Phi_z(t)}( t^*_{z}-t)\big]\int_E K_{z^-}(z^+)U^*(z^+,s+ t^*_{z})d\nu_{z^-}(z^+) \\
 &\mbox{where } z^-=\Phi_{z}(t^*_z)
 \end{align*}
so $ U^{*}\big(\Phi_{z}(t),s+t\big) =o(1) +(1+o(1))U^{\text{-}}\big(\Phi_{z}(t_z^{*}),s+t_z^{*}\big) $ as $t\to t^*_z,\ t<t^*_z$. 
 \end{proof}
 
 \subsection{Proof of Theorem \ref{th:optijumprate}} 
 \label{appendix-sec22}
 \begin{proof}
 We have seen in the proof above that
 \begin{align*}
 U^{*}\big(\Phi_{z}(t),s+t\big)
 &=\int_t^{t^*_{z}} U^-( \Phi_{z}(u),s+ u)\lambda_{z}(u) \exp\big[-\Lambda_{\Phi_z(t)}(u-t)\big] du \\&\quad + \exp\big[-\Lambda_{\Phi_z(t)}( t^*_{z}-t)\big] \int_E K_{z^-}(z^+)U^*(z^+,s+ t^*_{z})d\nu_{z^-}(z^+) \end{align*}
 so
 \begin{align*}
U^{*}\big(\Phi_{z}(t),s+t\big)&=\int_t^{t^*_{z}} U^-( \Phi_{z}(u),s+ u)\lambda_{z}(u) \exp\big[-\Lambda_{z}(u )\big] \exp\big[+\Lambda_{z}(t )\big]du \\
 &\quad + \exp\big[-\Lambda_{z}( t^*_{z} )\big] \exp\big[+\Lambda_{z}(t )\big] \int_E K_{z^-}(z^+)U^*(z^+,s+ t^*_{z})d\nu_{z^-}(z^+) \\ 
 &=\frac 1 {\exp\big[-\Lambda_{z}(t)\big]}\int_{[t,t^*_{z}]} U^-( \Phi_{z}(u),s+ u)\Big(\lambda_{z}(u)\Big)^{\indic{t< t^*_z}} \exp\big[-\Lambda_{z}(u )\big] d\mu_z(t) 
 \end{align*}
 This last equality allows to transform \eqref{LB1} into \eqref{LB2}.
 \end{proof}

 \subsection{Equality \eqref{Uderiv}}
 \label{appendix-sec23}
 Let $z\in E$ and $s\in[0,t_f)$. Remember that equality \eqref{Uderiv} states that if the functions \linebreak $u\to U^{\text{-}}\big(\Phi_{z}(u),s+u\big) $ and $u\to \lambda_{ z}(v)$ are continuous almost everywhere on $[0,t^*_z)$, then almost everywhere $$ \frac{\partial U^{*}\big(\Phi_{z}(v),s+v\big) }{\partial v}= U^{*}\big(\Phi_{z}(v),s+v\big) \lambda_{ z}(v) - U^{\text{-}}\big(\Phi_{z}(v),s+v\big) \lambda_{ z}(v)$$
\begin{proof} 
 We denote by $T$ the time until the next jump after the trajectory has reached $Z_s=z$. For $0\leq h<t^*_z$, we define $\tau=\min(h,T)$.
 \begin{align*}
 U^{*}(z,s)&=\mathbb E\big[\indic{\mathscr D}(\mathbf{Z})\big|Z_s=z\big]\\
 &=\mathbb E\Big[\mathbb E\big[\indic{\mathscr D}(\mathbf{Z})\big|Z_{s+\tau}\big]\Big|Z_s=z\Big]\\
 &=\mathbb E\Big[(\indic{\tau=h}+\indic{\tau<h})\mathbb E\big[\indic{\mathscr D}(\mathbf{Z})\big|Z_{s+\tau}\big]\Big|Z_s=z\Big]\\
 &=\mathbb E\Big[ \indic{T=h} \,\mathbb E\big[\indic{\mathscr D}(\mathbf{Z})\big|Z_{s+h}=\Phi_z(h)\big]\Big|Z_s=z\Big]\ +\ 
 \mathbb E\Big[ \indic{T<h}\, \mathbb E\big[\indic{\mathscr D}(\mathbf{Z})\big|Z_{s+T}\big]\Big|Z_s=z\Big]\\
 &= U^{*}(\phi_z(h),s+h) \ \mathbb E \big[ \indic{T=h} \big|Z_s=z\ \big]\ +\ 
 \mathbb E\Big[ \indic{T<h}\, U^{*}(Z_{s+T},s+T)\Big|Z_s=z\Big]\\[10pt]
 &=U^{*}(\phi_z(h),s+h) \,\exp\big[-\Lambda_z(h)\big] +\ \int_0^h \int_E K_{\Phi_z(u)}(z^+) U^{*}(z^+,s+u)d\nu_{\Phi_z(u)}(z^+)\lambda_z(u)\exp\big[-\Lambda_z(u)\big] du \end{align*}
 As $\lambda_z(.)$ is continuous almost everywhere we have that almost everywhere :
 \begin{align*}
 U^{*}(z,s)&=U^{*}(\phi_z(h),s+h) \,(1-\lambda_z(0)h+o( h )) +\ \int_0^h U^{\text{-}}(\Phi_z(u),s+u) \lambda_z(u)\exp\big[-\Lambda_z(u)\big] du
 \end{align*} As $u\to U^{\text{-}}(\phi_z(u),s+u)\lambda_z(u)$ is continuous almost everywhere, and we can do a Taylor approximation of the integral, which gives :
 \begin{align*}
 U^{*}(z,s)-U^{*}(\phi_z(h),s+h) &= -\lambda_z(0)\,.h\,.U^{*}(\phi_z(h),s+h) \, 
 +\ h\,.U^{\text{-}}(z,s) \lambda_z(0) +o(h ) 
 \end{align*} So $u\to U^{*}(\phi_z(u),s+u)$ is right-continuous almost everywhere. Therefore $U^{*}(\phi_z(h),s+h) =U^{*}(z,s) +o(1)$, and we get :
 \begin{align*}
 \frac{U^{*}(z,s)-U^{*}(\phi_z(h),s+h)}{h}&= -\lambda_z(0)\, U^{*}(z,s ) \, 
 +\ U^{\text{-}}(z,s) \lambda_z(0) +o(1 )
 \end{align*} Making $h$ tends to zero we get that $ u\to U^{*}(\phi_z(u),s+u)$ has a right-derivative in zero. Applying the same kind of reasoning in state $\Phi_z(-h)$ instead of $z$, we would find that the left-derivative exists and is equal to the right-derivative. So for almost every state $ z\in E$,
 $$ \bigg(\frac{\partial U^{*}\big(\Phi_{z}(v),s+v\big) }{\partial v}\bigg)_{v=0}= U^{*}\big(\Phi_{z}(0),s+0\big) \lambda_{ z}(0) - U^{\text{-}}\big(\Phi_{z}(0),s+0\big) \lambda_{ z}(0) $$Applying the same reasoning in a state $\Phi_{z_o}(v)$ instead of $z$ and using the additivity of the flow, we get that almost everywhere:

 $$ \forall z_o\in E, v>0,\quad \frac{\partial U^{*}\big(\Phi_{z_o}(v),s+v\big) }{\partial v} = U^{*}\big(\Phi_{z_o}(v),s+v\big) \lambda_{ z_o}(v) - U^{\text{-}}\big(\Phi_{z_o}(v),s+v\big) \lambda_{ z_o}(v) $$ 
\end{proof}

\newpage
{
\footnotesize
\bibliographystyle{apt}
\bibliography{sample}
}
\end{document}